\documentclass[11pt,onecolumn,draftclsnofoot]{IEEEtran}

\usepackage{xcolor}
\usepackage{amsfonts}
\usepackage{amssymb}
\usepackage{mathtools}
\usepackage[shortlabels]{enumitem}
\usepackage{hyperref}

\usepackage{tikz}
\usetikzlibrary{arrows,positioning,calc,fit,shapes.geometric}

\usepackage{amsthm} 
\newtheorem{theorem}{Theorem}[section]
\newtheorem{proposition}[theorem]{Proposition}
\newtheorem{definition}[theorem]{Definition}

\newtheorem{lemma}[theorem]{Lemma}
\newtheorem{example}{Example}[section]
\newtheorem{nonexample}[example]{(Non--)Example}
\newtheorem{corollary}[theorem]{Corollary}

\makeatletter
\newenvironment{proofsketch}[1][Proof sketch]
{\begin{proof}[#1]}{\end{proof}}
\makeatother
\newenvironment{silentproof}
  {\par\pushQED{\qed}\noindent\ignorespaces}
  {\popQED\par}

\newcommand{\mc}[1]{\mathcal{#1}}
\usepackage{mathrsfs}
\newcommand{\ms}[1]{\mathscr{#1}}
\newcommand{\NN}{\mathbb{N}}

\newcommand{\ZZ}{\mathbb{Z}}

\newcommand{\argmax}[1]{\underset{#1}{\operatorname{argmax}}}
\newcommand{\ov}[1]{\overline{#1}}

\newcommand{\pa}[1]{\mathrm{pa}(#1)}
\newcommand{\ch}[1]{\mathrm{ch}(#1)}
\newcommand{\pai}[2]{\mathrm{pa}_{#1}(#2)}

\usepackage[utf8]{inputenc} 
\usepackage[T1]{fontenc}
\usepackage{url}
\usepackage{ifthen}
\usepackage{cite}

\begin{document}
\title{Poset-Markov Channels: \\Capacity via Group Symmetry
}
\author{Eray Unsal Atay, Eitan Levin, Venkat Chandrasekaran, Victoria Kostina
\thanks{
This work was supported in part by the National Science Foundation (NSF) under grant~CCF-1956386, by AFOSR grants FA9550-23-1-0070 and FA9550-23-1-0204, and by the Carver Mead New Adventures Fund.
The authors are with the California Institute of Technology (emails: \{eatay,eitanl,venkatc,vkostina\}@caltech.edu).
}
}

\maketitle


\begin{abstract}
    Computing channel capacity is in general intractable because it is given by the limit of a sequence of optimization problems whose dimensionality grows to infinity.
    As a result, constant-sized characterizations of feedback or non-feedback capacity are known for only a few classes of channels with memory.
    This paper introduces \emph{poset-causal channels}---a new formalism of a communication channel in which channel inputs and outputs are indexed by the elements of a partially ordered set (poset).
    We develop a novel methodology that allows us to establish a single-letter upper bound on the feedback capacity of a subclass of poset-causal channels
    whose memory structure exhibits a Markov property
    and symmetry.
    The methodology is based on symmetry reduction in optimization.
    We instantiate our method on two channel models: the Noisy Output is The STate (NOST) channel---for which the bound is tight---and a new two-dimensional extension of it.
\end{abstract}

\section{Introduction} \label{intro}

Shannon’s original work~\cite{Shannon1948} established \emph{channel capacity} as the fundamental limit on the rate of reliable communication over a noisy channel. In addition to characterizing the fundamental rate limit, channel capacity motivates constructions of coding schemes that can be implemented in practical communication settings, such as low-density parity-check (LDPC) codes~\cite{Gallager1962}, turbo codes~\cite{Berrou1993}, and polar codes~\cite{Arikan2009}.
Consequently, studying the capacity of a channel remains a central theme in information theory.

\emph{Feedback channel capacity}, which characterizes the maximal rate of reliable communication over the channel when
feedback is present from the channel outputs to the encoder,
is of independent interest
due to its relation with real-time settings---feedback capacity characterizations lead to constructive, real-time coding schemes
that are of particular relevance to remote control~\cite{Khina2019, Han2023, Han2024}.
In particular,
the posterior matching principle~\cite{Shayevitz2011posterior},
which arises from the mutual-information capacity formulation with feedback,
achieves capacity for a broad class of memoryless channels.
Posterior matching crystallizes the principle behind the Horstein scheme for the binary symmetric channel (BSC)~\cite{Horstein1963} and the Elias scheme for the additive white Gaussian noise (AWGN) channel~\cite{Elias1956}---that of shaping the input to the channel to be most informative given the information that the receiver already has.

For general channels with memory, both the feedback capacity and the Shannon (non-feedback) capacity are given by multi-letter expressions---an optimization over $n$-symbol input sequences whose limiting value
equals the capacity (see~\cite{Dobrushin1963} for Shannon capacity and~\cite{Massey1990,Tatikonda2009} for feedback capacity).
It has been shown that no algorithmic procedure (Turing machine) can, in general, take a finite description of such a channel and compute these capacity limits (or even approximate them to arbitrary precision); in computability-theoretic terms, the Shannon and feedback capacity functions are \mbox{\emph{non-computable}~\cite{Boche2020FSC, Grigorescu2024feedback}.}
Accordingly,
exact capacity characterizations are known only for a limited number of channels~\cite{Butman, CoverPombra1989, Elia, Kim-Stationary-Gaussian, Gattami-Gaussian, Sabag-MIMO-Gaussian, trapdoor, Ising, permuter2014post, NOST, huleihel2024capacity, Sabag-no-consec-ones, Sabag-BIBO, Peled-0-k-RLL, alajaji1994effect, alajaji1995feedback, song2018capacity},
and computable estimates, in the form of upper and lower bounds,
are of interest.

While lower bounds (or achievability results) are usually proven through coding schemes~\cite{Shannon1948, CoverThomas}, upper bounds (or converse results) are often challenging because proving that the rate of communication cannot exceed a certain quantity must essentially rule out all possible coding schemes.
Such converse results are particularly difficult to derive in channel settings where memory is present~\cite{Mushkin1989, trapdoor, Huleihel-upper-bound}.

The feedback capacity of any given channel serves as an upper bound on its Shannon capacity, since allowing feedback
can only increase the achievable communication rate.
Thus, an upper bound on the feedback capacity implies an upper bound on the Shannon capacity---an approach that was taken in earlier works~\cite{Yang2005, Sabag-upper-bound}.

In this paper, we develop upper bounds on the feedback capacity for a novel channel model we term \emph{poset-causal channels}
whose inputs and outputs are indexed by the elements of a partially ordered set (poset), thus generalizing the classical one-dimensional time-indexed channel,
and whose memory structure exhibits a Markov property
and symmetry.
These upper bounds yield converse results for both feedback and Shannon capacities.

Throughout the paper,
uppercase letters denote random variables (e.g.,~$X,Y$),
and
lowercase letters denote realizations of random variables (e.g.,~$x,y$).
For positive integers $n$, we use the notations\linebreak ${X^n = (X_1 , \ldots , X_n)}$, $x^n = (x_1 , \ldots , x_n)$.
For random variables $X,Y$ we denote the probability of the event $X=x,Y=y$ by $P_{X,Y} (x,y)$,
and the conditional probability of $X=x$ given $Y=y$ by $P_{X \mid Y} (x \mid y)$.
Calligraphic letters denote finite sets (e.g.,~$\mc X, \mc Y$) and $|\cdot|$ denotes their cardinalities (e.g.,~$|\mc X|, |\mc Y|$),
whereas infinite sets are denoted by script letters (e.g.,~$\ms X, \ms Y$).
For positive integers $n$, we denote the set $\{1,\ldots,n\}$ by $[n]$.

\subsection{Background}

A discrete-time point-to-point communication channel with inputs $X_1,X_2,\ldots$ and outputs $Y_1,Y_2,\ldots$
is specified by a sequence of conditional probability distributions
\begin{equation}
    P_{ Y_t \mid X^t, Y^{t-1} } , \ t = 1, 2, \ldots ,
\end{equation}
Under mild regularity conditions, the feedback capacity of this channel is given by~\cite{Massey1990, Tatikonda2009}
\begin{equation} \label{eq:C-lim}
    C^{\textup{fb}} = \lim_{n \to \infty} \max_{\left\{ P_{X_t|X^{t-1},Y^{t-1}} \right\}_{t=1}^n} \frac{1}{n} I(X^n \to Y^n)
\end{equation}
where the \emph{directed information} $I(X^n \to Y^n)$ is given by~\cite{Massey1990}
\begin{equation}
    I(X^n \to Y^n) \coloneq \sum_{t=1}^n I(X^t ; Y_t \mid Y^{t-1}) ,
\end{equation}
and the $\max$ is over the \emph{channel input distributions} $P_{X_t \mid Y^{t-1}, X^{t-1}}$. Note that, together with the channel constants, the input distributions determine the joint distribution
\begin{equation} \label{eq:joint-induced}
    P_{X^n,Y^n} = \prod_{t=1}^n P_{ Y_t \mid X^t, Y^{t-1} } \, P_{X_t \mid X^{t-1}, Y^{t-1}} .
\end{equation}

The capacity expression in~\eqref{eq:C-lim} is the limit of a sequence of optimization problems. Even though each problem in the sequence is a concave maximization problem, and hence
in principle
computationally tractable, the dimensionality of the problems in the sequence grows with $n$, rendering the computation of the limit intractable in general~\cite{Boche2020FSC, Grigorescu2024feedback}. In this paper, we present a new methodology for overcoming this challenge by leveraging group symmetry in the underlying optimization problems.

As our point of departure, we consider the classical example of the stationary discrete memoryless channel (DMC), whose channel constants satisfy
\begin{equation} \label{DMC-stationary}
    P_{Y_t \mid X^t, Y^{t-1}} = Q_{Y_1 | X_1} , \ t = 1, 2, \ldots \, .
\end{equation}
Since the DMC is memoryless, the $n$\textsuperscript{th} maximization problem on the right-hand side of~\eqref{eq:C-lim} is equivalent to~\cite{Massey1990}
\begin{align}
    \max_{ \left\{ P_{X_t} \right\}_{t=1}^n } \frac{1}{n}\sum_{t=1}^n I( X_t ; Y_t )
    =\, \max_{P_{X_1}} \, I(X_1; Y_1) . \label{DMC_final_step}
\end{align}
That is, the sequence of growing-dimensional optimization problems reduces to a constant-sized \mbox{single-letter} one.
Although the equality in~\eqref{DMC_final_step} is evident from~\eqref{DMC-stationary}, we can also
interpret the simplification in~\eqref{DMC_final_step} from the perspective of group symmetry.
First, observe that each $I(X_t ; Y_t)$ is a concave function of $P_{X_t}$ as each $P_{Y_t|X_t}$ is fixed~\cite[Thm.~2.7.4]{CoverThomas}, and therefore the objective function $\sum_{t=1}^n I( X_t ; Y_t )$ is a concave function of the whole set of variables $\{ P_{X_t} \}_{t=1}^n$.  Next, this objective function is invariant with respect to the permutation group, which acts by permuting the variables $\{ P_{X_t} \}_{t=1}^n$.
Taken together, these two observations yield the single-letterization~\eqref{DMC_final_step} as a consequence of a general result from convex optimization~\cite[Ex.~4.4]{Boyd2004convex}---in a convex program in which the objective and constraints are invariant under the action of a group (i.e., \emph{symmetric}), the constraint set may be
restricted to variables invariant under the group.
In the case of the DMC,
replacing each $P_{X_t}$ by $\frac{1}{n} \sum_{t=1}^n P_{X_t}$ and by appealing to the concavity,
this amounts to the restriction $P_{X_t}=P_{X_1}$ for all $t$, as in~\eqref{DMC_final_step}.
Motivated by this reinterpretation of~\eqref{DMC_final_step},
we use
symmetry reduction
in Theorem~\ref{thm:poset-single-letter}
below
to obtain a single-letter upper bound on capacities of a certain class of channels.

\subsection{Contributions}

In this work, we generalize
the usual notion of causality with respect to time index to a new notion of causality with respect to partial order, which we refer to as \emph{poset causality}.
The $n$-letter optimization problems whose limit gives the capacity are rarely symmetric beyond the simple example of the DMC. However,
we observe that they are \emph{approximately} symmetric for a large class of channels.
We propose a novel methodology for deriving single-letter convex upper bounds for capacities of
certain channels defined via poset causality.
Our methodology consists of formulating convex relaxations
of the
feedback capacity,
and deducing that these relaxations converge to a single-letter
convex
problem through the (approximate) symmetry they exhibit.

First, in Section~\ref{sec:NOST}, we present a novel and conceptually
transparent
derivation of the feedback capacity of the Noisy Output is the STate (NOST) channel~\cite{chen2005capacity, NOST} as a demonstration of our methodology.
The key component of our derivation is establishing the independence of capacity on the initial distribution in Lemma~\ref{lem:NOST-indep-init-dist} via a Markov-mixing argument based on the fixed-point theory of set-valued maps.
The independence with respect to the initial distribution then
yields
a simple proof of the achievability part in Proposition~\ref{prop:NOST-achievability}. 


Next, as motivation for our development of a broader framework, we consider a two-dimensional variant of the NOST channel.
Communication channels defined on two-dimensional grids~\cite{berger-2d-ising, ReRAM} have been studied previously, and the fact that the channel’s inputs and outputs are indexed by a pair of numbers (rather than by a single index such as time) drastically complicates the analysis and computation of channel capacity.
To formalize and address these difficulties,
in Section~\ref{sec:poset} we consider a class of channels with inputs and outputs indexed by elements of a poset,
the \emph{poset-causal channels}.
Under certain Markovianity
and approximate symmetry assumptions detailed in Definitions \ref{def:poset-causal}, \ref{def:poset-markov}, and \ref{def:approx-sym}, we show that the feedback capacity of such channels is upper-bounded by a single-letter concave maximization problem.

Our proof methodology for obtaining the single-letter upper bound on feedback capacity is as follows.
We start by formulating an upper bound on capacity in terms of a limit of maximizations where for each maximization the set of variables and the objective function are in terms of fixed-dimensional local subsets of inputs and outputs of the channel.
We observe that each maximization in this limit is a convex problem due to concavity of mutual information.
We then obtain a single-letter reformulation for
the limit of
these convex
problems
under further symmetry assumptions,
where we use the intuition and the principles of symmetry reduction in~\cite[Ex.~4.4]{Boyd2004convex}.
More specifically, when the channel's representation through a directed acyclic graph (DAG) is an induced subgraph of an acyclic Cayley graph, we show that the growing size optimization problems
upper-bounding
capacity converge to a single-letter problem.
This methodology results in a single-letter convex upper bound on the feedback capacity for
poset-causal channels
satisfying
a Markov property and approximate symmetry, given in Theorem~\ref{thm:poset-single-letter}.

\subsection{Related Work}

\paragraph*{Capacity characterizations}
The capacity of channels with memory and feedback are explicitly characterized in only a few special cases, namely, certain Gaussian channels with memory~\cite{Butman, CoverPombra1989, Elia, Kim-Stationary-Gaussian, Gattami-Gaussian, Sabag-MIMO-Gaussian} and certain finite-state channels (FSCs)~\cite{trapdoor, Ising, permuter2014post, NOST, huleihel2024capacity, Sabag-no-consec-ones, Sabag-BIBO, Peled-0-k-RLL, alajaji1994effect, alajaji1995feedback, song2018capacity}.
In particular, a single-letter convex optimization expression for the feedback capacity was previously known only for the multiple-input multiple-output (MIMO) Gaussian channel~\cite{Sabag-MIMO-Gaussian} and the Noisy Output is the STate (NOST) channel~\cite{NOST}.

\paragraph*{Gaussian channels}
The feedback capacity of channels with additive Gaussian noise
is derived as the limit of an $n$-letter matrix optimization problem~\cite{CoverPombra1989}. It is characterized by a finite-dimensional optimization when the Gaussian noise follows a state-space model~\cite{Kim-Stationary-Gaussian}. In~\cite{Gattami-Gaussian}, the idea of using linear matrix inequalities (LMIs) is proposed. The feedback capacity of a class of MIMO Gaussian channels is written as a finite-dimensional convex problem with LMI constraints in~\cite{Sabag-MIMO-Gaussian}.

\paragraph*{Finite-state channels}
An FSC is characterized by the conditional probabilities $Q_{Y, S | X, S^\prime}$, where $S^\prime$ and $S$ represent the previous and the current channel state, respectively.
A prominent example is the \mbox{trapdoor} channel introduced in~\cite{Blackwell1961}, which provides a simple model related to biological communication~\cite{berger2002lec}.
Its feedback capacity is formulated as an infinite-horizon dynamic program, the exact solution for which is explicitly found~\cite{trapdoor}.
Other FSCs with known feedback capacities are binary channels with run-length input constraints~\cite{Sabag-no-consec-ones, Sabag-BIBO, Peled-0-k-RLL}, FSCs with certain input-output \mbox{symmetries~\cite{alajaji1995feedback,alajaji1994effect,song2018capacity},} the Ising channel~\cite{Ising}, and the NOST channel~\cite{NOST}.
The NOST channel, introduced in~\cite{chen2005capacity}, can be characterized without the channel state $S$, through the conditional probabilities $Q_{Y | X, Y^\prime}$ where $Y^\prime$ represents the previous channel output. The feedback capacity of the NOST channel is given as a recursive program in~\cite{chen2005capacity}, and as the following convex optimization in~\cite{NOST}:

\begin{equation} \label{eq:C-NOST}
    C^{\textup{fb}}_{\textup{NOST}} \ = \ \max_{ P_{X,Y'} } I( X ; Y \mid Y' ) \quad \textup{s.t. } P_Y = P_{Y'}
\end{equation}
where $Y^\prime$ represents the previous channel output.
The constraint $P_Y = P_{Y^\prime}$ on the variables $P_{X,Y'}$
asserts that the channel output distribution is stationary when using the input distribution $P_{X|Y'}$.
In the \mbox{single-letter} capacity expression~\eqref{eq:C-NOST}
in~\cite{NOST}, the converse
is proven by upper-bounding the mutual information expression using a stationarity/ergodicity argument,
whereas the achievability is proven by constructing a coding scheme that performs rate-splitting.

\paragraph*{Methodologies}
A dynamic programming formulation of channel capacity is given in~\cite{Tatikonda2000thesis, Tatikonda2009} and pursued in~\cite{trapdoor, Sabag-no-consec-ones, Sabag-BIBO, Peled-0-k-RLL, Ising, chen2005capacity}. Input-output symmetry is utilized in 
\cite{alajaji1995feedback,alajaji1994effect,song2018capacity} to show that feedback cannot increase the capacity of the channels considered therein. A method to derive \mbox{single-letter} upper bounds on capacities of certain FSCs based on mapping output sequences to nodes on an auxiliary directed graph called the \emph{Q-graph} is proposed in~\cite{Sabag-upper-bound}.

\paragraph*{Information theory of random fields}
In a standard discrete-time communication channel,
the
input and output processes are indexed by
$\mathbb N$. Channels whose input and output processes are indexed by more general index sets are studied in the
literature on the
information theory of random fields literature~\cite{bonomi1985thesis}. An instance of such a communication channel, the 2D Ising channel~\cite{berger-2d-ising}, has inputs and outputs indexed by the points on 2D grid lattice.
Another instance is the resistive random-access memory (ReRAM) channel~\cite{ReRAM}, which models the readout of data stored in a square‑grid array of resistive cells,
so the physical storage medium behaves as a communication channel with memory.

\paragraph*{Other areas}
Computing the limiting optimal value of sequences of optimization problems of growing size is an important problem in several other applications, including extremal\linebreak \mbox{combinatorics~\cite{raymond2018symmetric,raymond2018symmetry,lovasz2012large,brosch_thesis},} quantum information~\cite{huber2021positive,klep2018positive,klep2022optimization}, and \mbox{mean-field} game theory~\cite{lasry2007mean}.
In all these applications, the optimization problems involved are group invariant,
and symmetry reduction often yields constant-sized programs~\cite{raymond2018symmetric,raymond2018symmetry,riener2013exploiting,levin2023free}. 
Thus, our paper represents a new point of contact between information theory and the above literature on group-invariant optimization.

\paragraph*{Organization}
In Section~\ref{sec:NOST}, we give a novel proof of the single-letter  feedback capacity characterization of the NOST channel.
In Section~\ref{sec:poset}, we introduce poset-causal channels and their subclasses of
poset-Markov and
approximately symmetric channels, illustrated with running examples.
We prove a single-letter convex upper bound on the feedback capacity of approximately symmetric
poset-Markov
channels.
In Section~\ref{sec:numerical}, we present numerical results.
We then conclude in Section~\ref{sec:conclusions}.

\section{NOST Channel} \label{sec:NOST}

One of the contributions of our work is a novel derivation of the single-letter
feedback capacity of the Noisy Output is the STate (NOST) channel~\cite{NOST}.
Unlike the original derivation in~\cite{NOST}—which relies on
a stationarity/ergodicity argument to restrict their analysis to stationary distributions—we work directly with the
limit expression~\eqref{eq:C-lim} and show that the same capacity emerges under a one-step positivity assumption that is stronger than the connectivity assumption in~\cite[Def.~1]{NOST}.
At a higher level, we use fixed-point theory to present a more transparent treatment of the initial distribution, which is effectively a ``seed'' to the channel.
This results in a simple proof of achievability.

The NOST channel is a discrete-time channel with inputs $\{X_t\}_{t\ge1}$ and outputs $\{Y_t\}_{t\ge0}$ taking values from their corresponding finite alphabets $\mc X$ and $\mc Y$.
The NOST channel is characterized by channel constants that satisfy
\begin{align}
    P_{Y_t\mid X^t,Y^{t-1}}
    &= P_{Y_t\mid X_t,Y_{t-1}} \label{eq:NOST-channel-const-1} \\
    &= Q_{Y\mid X,Y'} \label{eq:NOST-channel-const-2}
\end{align}
for all $t\ge1$ and for fixed conditional distributions $Q_{Y\mid X,Y'}$ that are
\emph{strictly positive}.\footnote{
The original work~\cite{NOST} requires the ``connectivity'' condition~\cite[Def.~1]{NOST}, namely the existence of a
\mbox{positive-probability} path in finitely many steps.
Our assumption is stronger (positivity in \emph{one} step) and
simplifies certain \mbox{Markov-mixing} arguments.
}
Equation~\eqref{eq:NOST-channel-const-1} states that, given the entire past,
$Y_t$ depends only on the current input $X_t$ and the previous output $Y_{t-1}$, as depicted in Figure~\ref{fig:NOST}
below.
Equation~\eqref{eq:NOST-channel-const-2} asserts that this dependence is
\emph{stationary}.

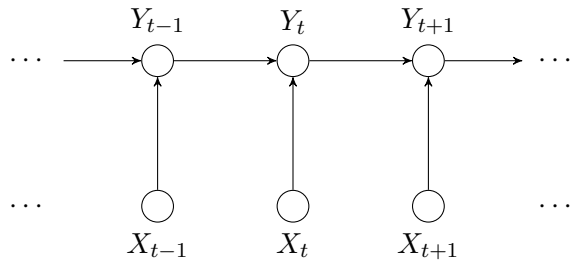
\begin{figure}[ht]
    \centering
    \begin{tikzpicture}[>=stealth', node distance=1.8cm]
        \tikzstyle{state}=[draw,circle,minimum size=12pt,inner sep=2pt]
        \node[state,label=below:$X_{t-1}$] (xprev) {};
        \node[state,right of=xprev,label=below:$X_{t}$] (xcur) {};
        \node[state,right of=xcur,label=below:$X_{t+1}$] (xnext) {};
        
        \node[state,above=1.5cm of xprev,label=above:$Y_{t-1}$] (yprev) {};
        \node[state,above=1.5cm of xcur,label=above:$Y_{t}$] (ycur) {};
        \node[state,above=1.5cm of xnext,label=above:$Y_{t+1}$] (ynext) {};
        
        \draw[->] (xprev) -- (yprev);
        \draw[->] (xcur)  -- (ycur);
        \draw[->] (xnext) -- (ynext);
        
        \draw[->] (yprev) -- (ycur);
        \draw[->] (ycur)  -- (ynext);

        \draw[->] ($(yprev)+(-1.25,0)$) -- (yprev);
        \draw[->] (ynext) -- ($(ynext)+(1.25,0)$);
        

        \node at ($(xprev.west)+(-1.5,0.0)$) {$\cdots$};
        \node at ($(xnext.east)+(1.5,0.0)$) {$\cdots$};
        \node at ($(yprev.west)+(-1.5,0.0)$) {$\cdots$};
        \node at ($(ynext.east)+(1.5,0.0)$) {$\cdots$};
    \end{tikzpicture}
    \caption{Depiction of the NOST channel input-output dependencies with current, previous, and next inputs and outputs.
    Each output $Y_t$ depends on the current input $X_t$ and the previous output $Y_{t-1}$, and is independent of the rest of the history given these two.
    }
    \label{fig:NOST}
\end{figure}

\paragraph*{Initial output $Y_0$}
We assume that the initial output $Y_0$ has a fixed distribution $P_0$, and is available both at the encoder and the decoder.
Section~\ref{subsec:NOST-indep} below shows—by an explicit
Markov-mixing argument—that the capacity is in fact independent of
the distribution of $Y_0$.
While~\cite{NOST} remarks on this independence, we
state it in a stronger form in Lemma~\ref{lem:NOST-indep-init-dist}
below,
and then in
Proposition~\ref{prop:NOST-max-over-init}
below,
we
incorporate
$P_0$
as an optimization variable.
Our goal in doing so is to obtain a more flexible characterization of the capacity than $Y_0$ having a fixed distribution.

Denote by $C^{\textup{fb}}_{\textup{NOST}}$ the feedback capacity~\eqref{eq:C-lim} of the NOST channel, characterized by the channel constants~\eqref{eq:NOST-channel-const-2}.
The next theorem restates the single-letter characterization of $C^{\textup{fb}}_{\textup{NOST}}$, proved \mbox{in~\cite[Thm.~1]{NOST}.}

\begin{theorem}[NOST single-letterization~\cite{NOST}] \label{thm:NOST}
    Feedback capacity of the NOST channel equals
    \begin{equation} \label{eq:NOST-C-single-let}
        C^{\textup{fb}}_{\textup{NOST}} \ = \ \max_{ P_{X,Y'} } I( X ; Y \mid Y' ) \quad \textup{s.t. } P_Y = P_{Y'} .
    \end{equation}
\end{theorem}

\noindent
We prove~\eqref{eq:NOST-C-single-let} in two main steps:

\begin{enumerate}
    \item \textbf{Independence with respect to initial distribution.}
        We show the independence of $C^{\textup{fb}}_{\textup{NOST}}$ with respect to the initial distribution
        $P_0$
        in Lemma~\ref{lem:NOST-indep-init-dist}
        below,
        and we add
        $P_0$
        as a variable to the optimizations in the capacity expression
        in Proposition~\ref{prop:NOST-max-over-init}
        below.
    
    \item \textbf{Single-letter achievability \& converse.}
        We prove the achievability and the converse in~\eqref{eq:NOST-C-single-let}, i.e., that the single-letter expression is an upper bound as well as a lower bound on $C^{\textup{fb}}_{\textup{NOST}}$.
        The achievability part
        in Proposition~\ref{prop:NOST-achievability}
        below
        is rather straightforward, as having $P_0$ as a variable
        makes it possible to
        replicate any feasible choice $P_{X,Y'}$ from the right-hand side on the left-hand side through an inductive construction, ensuring that the objective value on the left is at least as large.
        The converse part
        in Proposition~\ref{prop:NOST-converse}
        below
        uses the intuition and the principles of symmetry \mbox{reduction~\cite[Ex.~4.4]{Boyd2004convex}}
        through
        averaging the sequence of joint distributions and invoking concavity via Jensen's inequality to upper-bound the objective in the limit.
\end{enumerate}

\subsection{Capacity and its independence of the initial distribution}
\label{subsec:NOST-indep}

The first main step is to show that the
feedback capacity of the NOST channel
is independent of the initial distribution
$P_0$,
and then to obtain a characterization of the channel capacity where 
$P_0$
is a variable instead of
begin fixed.
In the proposition below, we simplify the capacity problem given in~\eqref{eq:C-lim}.

\begin{proposition}[NOST simplification] \label{prop:NOST-formulation}
    The feedback-capacity limit in~\eqref{eq:C-lim} for the NOST channel equals
    \begin{equation}
        \lim_{n \to \infty} \max_{\left\{ P_{X_t|X^{t-1},Y^{t-1}} \right\}_{t=1}^n} \frac{1}{n} I(X^n \to Y^n)
        \ = \
        \lim_{n \to \infty} \max_{ \left\{ P_{ X_t \mid Y_{t-1} } \right\}_{t=1}^n } \frac{1}{n} \sum_{t=1}^n I( X_t ; Y_t \mid Y_{t-1} )
    \end{equation}
    where $Y_0\sim P_0$.
\end{proposition}

\begin{proofsketch}
    We simplify each mutual information term in the directed information together with the corresponding input distribution variable.
    The simplified terms reflect the Markovianity of the NOST channel in~\eqref{eq:NOST-channel-const-1},
    as each contains $X_t,Y_t,Y_{t-1}$ only.
    The full proof can be found in Appendix~\ref{sec:proof-prop:NOST-formulation}.
\end{proofsketch}

The result in Proposition~\ref{prop:NOST-formulation} appears in~\cite[eq.~(42-b)]{NOST} as a simplification of the mutual information sum.
Prior to that, it appeared as an upper bound
for the converse proofs in~\cite[eq.~(6)]{chen2005capacity},~\cite[eq.~(36)]{permuter2014post}.
Here, we include it to ensure the completeness of our proof methodology.

Next, we formalize the independence of $C^{\textup{fb}}_{\textup{NOST}}$ with respect to the initial distribution
$P_0$.

\begin{lemma}[NOST independence of the initial distribution] \label{lem:NOST-indep-init-dist}
    For any sequence $\{\widetilde P^{(n)}_{Y_0}\}_{n\ge1}$ of initial distributions,
    \begin{equation} \label{eq:NOST-indep-init-dist}
        \begin{aligned}
            C^{\textup{fb}}_{\textup{NOST}}
            =
            \lim_{n \rightarrow \infty} \frac{1}{n} & \max_{ \left\{ P_{ X_t \mid Y_{t-1} } \right\}_{t=1}^n } \sum_{t=1}^n I( X_t ; Y_t \mid Y_{t-1} )
            =
            \lim_{n \rightarrow \infty} \frac{1}{n} \max_{ \left\{ P_{ X_t \mid Y_{t-1} } \right\}_{t=1}^n } \sum_{t=1}^n I( X_t ; Y_t \mid Y_{t-1} ) \\
            & \hspace{0.8cm} \textup{s.t.} \quad Y_0 \sim P_0
            \hspace{4.7cm} \textup{s.t.} \quad Y_0 \sim \widetilde{P}_0^{(n)} .
        \end{aligned}
    \end{equation}
\end{lemma}

\begin{proofsketch}
    The independence on the initial distribution essentially pertains to mixing of Markov chains.
    The proof relies on Banach fixed point theorem (a.k.a. contraction mapping theorem).
    We prove that the set of feasible values of $P_{Y_n}$ as $n\to\infty$ concentrates around a closed and bounded subset $K$ of the set of probability distributions over the output alphabet $\mc Y$,
    where the closed and bounded subset $K$ does not depend on
    $P_0$.
    The mappings under consideration are proven to be contraction mappings by the strict positivity of the channel constants $Q_{Y \mid X , Y'}$.
    The full proof can be found in Appendix~\ref{sec:proof-lem:NOST-indep-init-dist}.
\end{proofsketch}


As a corollary of Lemma~\ref{lem:NOST-indep-init-dist}, we show
that
$P_0$
can be
added
to the maximizations in~\eqref{eq:NOST-indep-init-dist} as a variable.

\begin{proposition}[NOST maximization over initial distribution] \label{prop:NOST-max-over-init}
    NOST capacity can be equivalently expressed by maximizing over the initial distribution
    $P_0$,
    i.e.,
    \begin{align} \label{eq:NOST-max-over-init}
        C^{\textup{fb}}_{\textup{NOST}}
        \ = \
        \lim_{n \rightarrow \infty} \frac{1}{n} \, & \max_{ P_0 \, , \, \left\{ P_{ X_t \mid Y_{t-1} } \right\}_{t=1}^n } \, \sum_{t=1}^n I( X_t ; Y_t \mid Y_{t-1} ) .
    \end{align}
\end{proposition}


\begin{proof}
    For each $n\ge1$, for the maximization on the right-hand side of~\eqref{eq:NOST-max-over-init}, denote the maximizing variables by
    \begin{equation} \label{eq:NOST-def-hat}
        \left( \widehat P^{(n)}_{Y_0} \, , \, \left\{ \widehat P^{(n)}_{ X_t \mid Y_{t-1} } \right\}_{t=1}^n \, \right)
        \ \coloneq \
        \argmax{ P_0 \, , \, \left\{ P_{ X_t \mid Y_{t-1} } \right\}_{t=1}^n } \, \sum_{t=1}^n I( X_t ; Y_t \mid Y_{t-1} ) .
    \end{equation}
    We can write the maximization on the RHS of~\eqref{eq:NOST-max-over-init} equivalently as
    \begin{equation} \label{eq:NOST-equiv-hat}
        \begin{aligned}
            \max_{ P_0 \, , \, \left\{ P_{ X_t \mid Y_{t-1} } \right\}_{t=1}^n } \, \sum_{t=1}^n I( X_t ; Y_t \mid Y_{t-1} )
            \ = \
            & \max_{ \left\{ P_{ X_t \mid Y_{t-1} } \right\}_{t=1}^n } \sum_{t=1}^n I( X_t ; Y_t \mid Y_{t-1} ) \\
            & \hspace{0.8cm} \textup{s.t.} \quad Y_0 \sim \widehat P^{(n)}_{Y_0} ,
        \end{aligned}
    \end{equation}
    which follows by the definition of $\widehat P^{(n)}_{Y_0}$ in~\eqref{eq:NOST-def-hat}. Notice that by~\eqref{eq:NOST-indep-init-dist}, the normalized limit of the RHS of~\eqref{eq:NOST-equiv-hat} equals the capacity $C^{\textup{fb}}_{\textup{NOST}}$.
   ~\eqref{eq:NOST-indep-init-dist} and~\eqref{eq:NOST-equiv-hat} imply~\eqref{eq:NOST-max-over-init}, as desired.
\end{proof}


The new characterization of NOST capacity in~\eqref{eq:NOST-max-over-init}, where the initial distribution $P_0$ is a variable, makes the analysis of the capacity more tractable than assuming a fixed initial distribution,
as the former includes more freedom with an additional variable. Intuitively, the characterization in~\eqref{eq:NOST-max-over-init} corresponds to a communication setting where the initial condition is \emph{user-configurable}, since effectively the initial distribution $P_0$ is an additional input.


\subsection{Capacity single-letterization}

We show that the expression for NOST capacity in~\eqref{eq:NOST-max-over-init}
admits the single-letter characterization
in~\eqref{eq:NOST-C-single-let}.
We accomplish this by proving the achievability and the converse. That is, we show that the single-letter expression in~\eqref{eq:NOST-C-single-let} lower bounds and also upper-bounds the expression on the right-hand side of~\eqref{eq:NOST-max-over-init}.

The proposition below establishes the achievability.

\begin{proposition}[NOST achievability] \label{prop:NOST-achievability}
    The single-letter expression in~\eqref{eq:NOST-C-single-let} lower-bounds NOST capacity, i.e., we have
    \begin{equation} \label{eq:NOST-achievability}
        \lim_{n \to \infty} \max_{\substack{ P_0 \\ \left\{P_{X_t \mid Y_{t-1}}\right\}_{t=1}^n }} \frac{1}{n} \sum_{t=1}^n I(X_t ; Y_t \mid Y_{t-1})
        \quad \ge \quad
        \max_{ P_{X,Y'} } I( X ; Y \mid Y' ) \quad \textup{s.t. } P_Y = P_{Y'} .
    \end{equation}
\end{proposition}

\begin{proofsketch}
    For any value of the objective function on the right-hand side, we construct a set of variables for the left-hand side with objective value at least as large.
    More explicitly,
    for any fixed feasible $P'_{X,Y'}$ for the right-hand side,
    having
    $P_0$
    as a variable on the left-hand side
    enables us to set $P_{X_t,Y_{t-1}} = P'_{X,Y'}$ for all $t$.
    This results in the left-hand side having the objective value of the right-hand side;
    thus, the lower bound follows.
    The full proof can be found in Appendix~\ref{sec:proof-prop:NOST-achievability}.
\end{proofsketch}

Next, the proposition below establishes the converse.

\begin{proposition}[NOST converse] \label{prop:NOST-converse}
    The single-letter expression in~\eqref{eq:NOST-C-single-let} upper-bounds NOST capacity, i.e., we have
    \begin{equation} \label{eq:NOST-converse}
        \lim_{n \to \infty} \max_{\substack{ P_0 \\ \left\{P_{X_t \mid Y_{t-1}}\right\}_{t=1}^n }} \frac{1}{n} \sum_{t=1}^n I(X_t ; Y_t \mid Y_{t-1})
        \quad \le \quad
        \max_{ P_{X,Y'} } I( X ; Y \mid Y' ) \quad \textup{s.t. } P_Y = P_{Y'} .
    \end{equation}
\end{proposition}

\begin{proofsketch}
    The proof makes use of the concavity of mutual information together with Jensen's inequality.
    Fix an $n>0$,
    take any feasible $P'_0 , \left\{P'_{X_t \mid Y_{t-1}}\right\}_{t=1}^n$ for the left-hand side, and consider the resulting joint distributions $\left\{P'_{X_t , Y_{t-1}}\right\}_{t=1}^n$ through the channel.
    Define the averaged-joint distribution $\widetilde P^{(n)}_{X , Y'}$ through
    \begin{equation}
        \widetilde P^{(n)}_{X , Y'} (x , y') \ \coloneq \ \frac{1}{n} \sum_{t=1}^n P'_{X_t,Y_{t-1}} (x , y') \quad \forall x,y' \in \mc X \times \mc Y .
    \end{equation}
    We show that $\widetilde P^{(n)}_{X , Y'}$ satisfies the constraint $P_Y = P_{Y'}$ \emph{in the limit} as $n \to \infty$.
    Hence, together with Jensen's inequality, we obtain the upper bound by taking $n\to\infty$.
    The main idea of this proof appeared in~\cite[Lemma~4]{NOST}.
    The full proof can be found in Appendix~\ref{sec:proof-prop:NOST-converse}.
\end{proofsketch}

Propositions~\ref{prop:NOST-achievability} and \ref{prop:NOST-converse} together imply
Theorem~\ref{thm:NOST}. As a result, the feedback capacity of NOST channel admits the single-letter characterization
\eqref{thm:NOST},
which matches~\cite[Thm.~1]{NOST}.

\section{General Framework} \label{sec:poset}

We define \emph{poset-causal channels} (Definition~\ref{def:poset-causal}) whose inputs and outputs are indexed by a poset,
and define the notion of feedback capacity for poset-causal channels.
If
a poset-causal channel satisfies a Markov property,
we call such channel a \emph{poset-Markov channel} (Definition~\ref{def:poset-markov}).
We establish an upper bound on the feedback capacity of poset-Markov channels
that is the limit of a growing-size sequence of convex problems.
Next, we consider a subclass of poset-Markov channels, \emph{approximately symmetric channels} (Definition~\ref{def:approx-sym}).
Our methodology leads to a single-letter upper bound on the feedback capacity of approximately symmetric channels in terms of the local behavior of each node given its neighbors, given in Theorem~\ref{thm:poset-single-letter}
below.
Throughout, we instantiate our results on the previously-studied NOST channel~\cite{chen2005capacity, NOST}, and its novel two-dimensional extension, the 2D NOST channel.

\subsection{Poset-Causal and Poset-Markov Channels} 

We define poset-causal channels by associating the nodes of a growing poset with channel input and output pairs,
and its edges with the transmission order.

\begin{definition}[Poset-causal channel]\label{def:poset-causal}

Let $(\ms V , \prec)$ be a countable poset.
A \emph{poset-causal channel} over $(\ms V, \ms E)$ is a collection of channel constants $\{P_{Y_v \mid X^{\preceq v}, Y^{\prec v}}\}_v$ for finite input and output alphabets $\mc X$ and $\mc Y$,
where $(\ms V, \ms E)$ has the following properties:
\begin{enumerate}[font=\textbf, align=left]
    \item[(Increasing sequence)]
    The set of nodes is an increasing union $\ms V = \bigcup_{n=1}^\infty \mc V_n$ of finite sets, so that $\mc V_n \subset V_{n+1}$
    and $\{w \mid w \prec v\} \subset \mc V_n$ for all $v \in \mc V_n$, and for all $n$.
    
    \item[(Initial and communication nodes)]
    The countable set of nodes $\ms V = \ms C \cup \ms I$ is the union of \emph{communication nodes}~$\ms C$ and \emph{initial nodes} $\ms I$. Each instance satisfies $\mc V_n = \mc C_n \cup \mc I_n$ where $\mc C_n = \mc V_n\cap \ms C$ and $\mc I_n = \mc V_n\cap \ms I$.
    
    \item[(Edge orientation)]
    There can be edges from the initial nodes in $\ms I$ to the communication nodes in $\ms C$ but not vice versa, i.e., $( \ms C \times \ms I ) \cap \ms E = \emptyset$.
    
\end{enumerate}
\end{definition}

In a poset-causal channel,
the communication nodes in $\ms C$ and the initial nodes in $\ms I$ have different characteristics pertaining to their roles in the communication.
Each initial node $v \in \ms I$ has an associated output $Y_v$, whereas each communication node $v \in \ms C$ has associated input and output $X_v$ and $Y_v$.
The partial order in the poset corresponds to causality relationships among the outputs.
The output at a given communication node $v \in \ms C$ depends only on the input at that node, $X_v$,
and on the collection of outputs
of the preceding nodes,
$Y^{\prec v}$.
Each communication node $v \in \ms C$ has feedback from the collection of outputs $Y^{\prec v}$.
We assume that the initial nodes $v \in \ms I$ have a fixed joint distribution $P_0$.
This characterizes the \emph{initial conditions} of the setting.
The results in the remainder of the paper hold for any fixed
$P_0$.

The poset $(\ms V , \prec)$ allows us to introduce a sequence of transmissions over the nodes $v \in \ms C$ as follows.
Extend each partial (sub-)order $(\mc C_n \setminus \mc C_{n-1} , \prec)$ to a total order,
which always exists by the order-extension principle~\cite{Szpilrajn1930},
and denote this total order by $\tau_n \colon \mc C_n \setminus \mc C_{n-1} \to \left| \mc C_n \setminus \mc C_{n-1} \right|$.
Next, concatenate the total orders $\{\tau_n\}_{n\ge1}$ to obtain a full total order $\tau \colon \ms C \to \NN$.
This ensures that $\mc C_n$ precedes $\mc C_{n+1}$ in the full total order $\tau$ for all $n$.
In particular, we have
\begin{equation} \label{eq:tau-subseq}
    \left\{ v \in \ms C \ \big| \ \tau(v) \le |\mc C_n| \right\} = \mc C_n \quad \forall n \ge 1.
\end{equation}
The full total order $\tau$ is precisely the order of transmissions.
The transmissions are governed by the channel constants
$P_{ Y_v \mid X^{\preceq v} , Y^{\prec v} }, v \in \ms C$.

The feedback channel setting and the corresponding feedback capacity are as follows.
The above \mbox{poset-causal} channel is a point-to-point communication channel ${(X_v)_{v \in \ms C} \longrightarrow (Y_v)_{v \in \ms C}}$.
To transmit a message over $n$ \emph{transmission stages}, the encoder maps it to the sequence of channel inputs $\{X_v\}_{v\in \mc C_n}$, and the decoder interprets the corresponding outputs $\{Y_v\}_{v\in \mc C_n}$ in an attempt to reconstruct the message, in the following manner.
Let $M$ denote the message size.
For a message $W \in [M] \coloneq \{1,\ldots,M\}$,
an encoder applies a sequence of mappings $(f_v)_{v \in \mc C_n}$ as inputs to nodes in $\mc C_n$, utilizing at stage $t$ the nodes in $\mc C_{t} \setminus \mc C_{t-1}$.
More specifically, at a given node $v \in \mc C_{t} \setminus \mc C_{t-1}$, the encoder applies
function $f_v \colon [M] \times \mc Y^{\prec v} \to \mc X$ that maps the original message $W \in [M]$ and the previous channel outputs $Y^{\prec v}$ to the input $X_v$ at node $v$.
The decoding happens after the $n$ stages when the decoder applies the decoding function $g \colon \mc Y^{\mc C_n} \to [M]$ to the outputs of the nodes in $\mc C_n$, outputting $\hat W \in [M]$, the estimate of the transmitted message $W$.
For an $\epsilon>0$, the pair $\left( \{f_v\}_{v \in \mc C_n} \, , \, g \right)$ is called an \emph{$(n,M,\epsilon)$--code} if the error probability satisfies $\mathbb P \left[ W \neq \hat W \right] \leq \epsilon$.
Denote the maximal achievable message size compatible with error probability $\epsilon$ after $n$ stages as
\begin{equation}
    M^*(n, \epsilon) \coloneq \max \{M \mid \exists \, (n, M, \epsilon) \text{--code}\} .
\end{equation}
The feedback capacity $C^{\textup{fb}}$ of the poset-causal channel is
\begin{equation} \label{eq:poset-c-def}
    C^{\textup{fb}} 
    \coloneq \lim_{\epsilon \to 0} \, \liminf_{n \to \infty} \, \frac{1}{|\mc C_n|} \log M^*(n, \epsilon) .
\end{equation}

In the example below, we demonstrate that any classical discrete-time channel can be instantiated as a poset-causal channel.

\begin{example}[Discrete-time channels]
    Take an arbitrary discrete-time channel, and assume that it is characterized by channel constants $\{ P_{ Y_n | X^n, Y^{n-1} } \}_{n\ge1}$.
    The set of nodes is the nonnegative integers $\ms V = \ZZ_{\ge0}$.
    Define the partial order as the usual integer order, i.e., ${m \prec n \iff m < n}$.
    Let ${\mc I = \{0\}}$, $\mc C = \{n\}_{n\ge1}$, and $V_n = \{0,1,\ldots,n\}$.
    The channel constants
    $\left\{P_{ Y_v \mid X^{\preceq v} , Y^{\prec v} }\right\}_{v \in \ms C}$
    become precisely $\{ P_{ Y_n | X^n, Y^{n-1} } \}_{n\ge1}$.
\end{example}
\noindent As a result, poset-causal channels subsume classical discrete-time channels.

The next definition pertains to poset-causal channels satisfying a Markov property.
Unlike poset-causal channels,
which are defined through a poset only,
poset-Markov channels are defined through a poset as well as a DAG. 
More specifically,
the partial order is induced by the DAG through the directed paths in the DAG, as follows.

\begin{definition}[Poset-Markov channel] \label{def:poset-markov}
    Let $(\ms V , \ms E)$ be a countable connected DAG.
    The DAG $(\ms V , \ms E)$ induces a partial order
    $\prec$ on $\ms V$ where $v \prec w$ if there exists a directed path from $v$ to $w$.
    Let $\pa v$ and $\ch v$ denote the set of parents and the set of children of $v$ in $(\ms V , \ms E)$, respectively.
    The poset-causal channel associated with the DAG $(\ms V , \ms E)$ and the poset $(\ms V , \prec)$ is a \emph{poset-Markov channel} if the channel constants $\{P_{Y_v \mid X^{\preceq v}, Y^{\prec v}}\}_v$
    satisfy $P_{ Y_v \mid X^{\preceq v} , Y^{\prec v} } = P_{ Y_v \mid X_v , Y_{\pa v} }$ for all $v \in \ms C$.
    In other words,
    for each communication node $v \in \ms V$,
    the output $Y_v$ depends only on the input $X_v$ and the outputs $Y_{\pa v}$.
\end{definition}

An infinite DAG is depicted in Figure~\ref{fig:dag} with initial and communications nodes indicated with different symbols.

\begin{figure}[ht]
\centering
\begin{tikzpicture}[>=stealth',scale=1,
    init/.style={diamond,draw,fill=gray!20,minimum size=20pt,inner sep=2pt},
    comm/.style={circle,draw,minimum size=20pt,inner sep=2pt}]

\foreach \i in {0,...,7}{
   \pgfmathsetmacro{\x}{1.5*\i}
   \ifnum\i<3
      \node[init] (n\i) at (\x,0) {};     
   \else
      \node[comm] (n\i) at (\x,0) {};     
   \fi
}

\draw[->] (n3) -- (n4);
\draw[->] (n4) -- (n5);

\draw[->,bend left=45]  (n1) to (n4);
\draw[->,bend left=45]  (n2) to (n6);
\draw[->,bend left=40]  (n3) to (n5);
\draw[->,bend left=40]  (n5) to (n7);

\draw[->,bend right=45] (n0) to (n3);   
\draw[->,bend right=40] (n4) to (n6);   

\node at ($(n7)+(1.0,0)$) {$\cdots$};

\end{tikzpicture}
\caption{Depiction of an infinite DAG. $\diamond$ Initial nodes. $\circ$ Communication nodes.}
\label{fig:dag}
\end{figure}
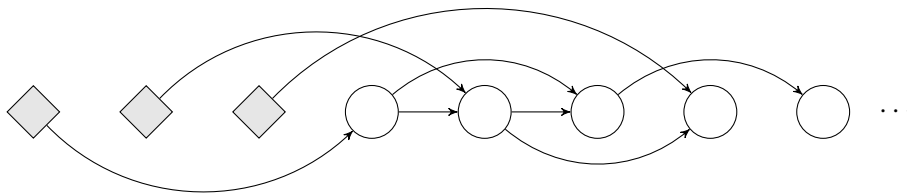

Below, we provide examples (and non-examples) of
poset-Markov channels.

\begin{example}[NOST channel]\label{ex:NOST-1}
    The NOST channel~\cite{NOST}
    is a
    \mbox{poset-Markov}
    channel.
    Here the set of nodes $\ms V=\ZZ_{\geq0}$ consists of nonnegative integers, with edges
    $\ms E=\{(n,n+1) \mid n\in\mathbb{Z}_{\geq0}\}$,
    so the partial order on $\ms V$ is the usual integer order.
    The set of initial nodes and communication nodes are $\ms I = \{0\}$ and ${\ms C = \{n\}_{n\ge1}}$, respectively.
    Each finite set $\mc V_n$ in the increasing sequence of sets is given by ${\mc V_n = \{0,1,\ldots,n\}}$.
    The channel constants satisfy Markovianity in the sense that
    ${P_{Y_n|X^n,Y^{n-1}}=P_{Y_n|X_n,Y_{n-1}}}$, so the output at the current step $n$ depends only on the current input and the previous output.
\end{example}

\begin{figure}[ht]
    \centering
    \begin{tikzpicture}[scale=2,>=stealth',
      init/.style={diamond,draw,fill=gray!25,minimum size=22pt,inner sep=1pt},
      comm/.style={circle,draw,minimum size=20pt,inner sep=2pt}]
    
        \node[init] (n0) at (0,0) {0};
        
        \foreach \i in {1,...,3}{
            \node[comm] (n\i) at (\i*1.2,0) {\i};
        }
        
        \foreach \i in {0,...,2}{
            \draw[->] (n\i) -- (n\the\numexpr\i+1\relax);
        }
        
        \coordinate (cont) at ($(n3)+(1.2,0)$);
        \draw[->] (n3) -- (cont);
        \node at ($(cont)+(0.35,0)$) {$\cdots$};
    \end{tikzpicture}
    \caption{Infinite DAG corresponding to the NOST channel in Example~\ref{ex:NOST-1}.
    Node $0$ is the initial node;
    the rest
    are communication nodes.}
    \label{fig:NOST-line}
\end{figure}
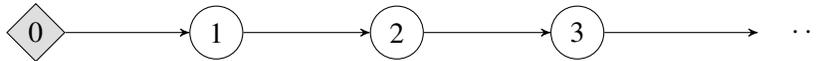

The DAG corresponding to the NOST channel is depicted in Figure~\ref{fig:NOST-line}.
More generally, the directed edges on the set of nodes $\ms V = \ZZ_{\ge0}$ can be defined to instantiate discrete-time channels that are \mbox{\emph{input-memoryless},} i.e., for each output $Y_n$, the input it depends on is only the current input $X_n$.
Such channels have channel constants of the form $\left\{ P_{Y_n \mid X_n , Y^{n-1}} \right\}_{n\ge1}$.

Recall the definition of
poset-Markov channels in Definition~\ref{def:poset-markov}.
The equality\linebreak $P_{ Y_v \mid X^{\preceq v} , Y^{\prec v} } = P_{ Y_v \mid X_v , Y_{\pa v} }$ for all $v \in \ms C$ asserts that each output $Y_v$ depends on a single\linebreak input---the input $X_v$ at the same node.
As a result, discrete-time channels that have input memory are not poset-Markov channels.
To better illustrate this limitation of poset-Markov channels,
below we provide an example of a channel that has input memory, and therefore is \emph{not} a channel that can be instantiated as a poset-Markov channel.

\begin{nonexample}[NOST channel with input memory]
    Consider a variant of the NOST\linebreak channel where each channel constant, instead of $P_{Y_n|X^n,Y^{n-1}}=P_{Y_n|X_n,Y_{n-1}}$, satisfies\linebreak $P_{Y_n|X^n,Y^{n-1}}=P_{Y_n|X_n,X_{n-1},Y_{n-1}}$.
    In other words, each output depends on the most recent two inputs rather than only one.
    
\end{nonexample}

The binary trapdoor channel~\cite{trapdoor},
because each transmission has an implicit dependence on all the input history,
cannot be instantiated as a poset-Markov channel.

\begin{nonexample}[Binary trapdoor channel~\cite{trapdoor}] \label{ex:trapdoor}
    The binary trapdoor channel produces outputs $Y_t$, where $Y_t$ is chosen uniformly at random from the set $\{X_t, S_{t-1}\}$ for each $t$. The state evolves according to $S_t = S_{t-1} \oplus X_t \oplus Y_t$.
\end{nonexample}
\noindent A characterization of the trapdoor channel solely in terms of the inputs and outputs that adheres to the Markovianity condition in Definition~\ref{def:poset-markov} does not exist.
In particular,
marginalizing away the state variable
results in the $n$\textsuperscript{th} output $Y_n$ depending on all past inputs and outputs.
Consequently, the trapdoor channel is \emph{not} a poset-Markov channel.


Every poset-causal channel is operationally equivalent to a
classical time-indexed
with no input memory.
One such equivalent
time-indexed channel
can always be constructed via the total order $\tau : \ms C \to \NN$ of transmissions.
However, the capacity
upper bound
we present in this section is most interesting for poset-Markov channels based on sparse DAGs. 
We proceed to give further examples of such DAGs.

\begin{example}[2D NOST channel]\label{ex:2DNOST-1}
    Another instantiation of a
    poset-Markov channel
    is via
    nodes ${\ms V=\ZZ_{\ge0}^2}$,
    and edges
    \begin{equation}
        \ms E=\{\,(\,(i,j),(i,j+1)\,)\,,\,(\,(i,j),(i+1,j)\,) \, \mid \, i,j\in\ZZ_{\ge0}\,\}.
    \end{equation}
    The induced partial order on $\ms V$ posits $(i,j) \preceq (k,l)$ if and only if $i\leq k$ and $j\leq l$.
    The set of initial nodes is the infinite set $\ms I=\{(i,j) \mid i=0 \textrm{ or } j=0\}$, and the communication nodes are $\ms C=\NN^2$.
    Each $\mc V_n$ is the finite square grid $\mc V_n = \{ (i,j) \mid 0 \le i,j \le n \}$.
    The channel constants are defined via $Q_{Y \mid X , Y'_1 , Y'_2}$ as
    $P_{Y_{i,j}|X^{\preceq (i,j)},Y^{\prec(i,j)}} = P_{Y_{i,j}|X_{i,j},Y_{i-1,j},Y_{i,j-1}} = Q_{Y \mid X , Y'_1 , Y'_2}$, so the output at each coordinate $(i,j)$ depends only on the input at coordinate $(i,j)$ and the outputs at coordinates $(i-1,j)$ and $(i,j-1)$.
\end{example}
\vspace{-0.5cm}
\begin{figure}[ht]
\centering
\begin{tikzpicture}[>=stealth',
  init/.style={diamond,draw,fill=gray!25,
               minimum size=12pt,inner sep=1.5pt},
  comm/.style={circle,draw,minimum size=12pt,inner sep=1.5pt}]

\foreach \x in {0,...,3}{
  \foreach \y in {0,...,3}{
    \ifnum\x=0\relax 
      \node[init] (n\x\y) at (\x,\y) {};
    \else\ifnum\y=0\relax 
      \node[init] (n\x\y) at (\x,\y) {};
    \else
      \node[comm] (n\x\y) at (\x,\y) {};
    \fi\fi
  }
}

\foreach \x in {0,...,2}{
  \foreach \y in {0,...,3}{
    \draw[->] (n\x\y) -- (n\the\numexpr\x+1\relax\y);
  }
}

\foreach \x in {0,...,3}{
  \foreach \y in {0,...,2}{
    \draw[->] (n\x\y) -- (n\x\the\numexpr\y+1\relax);
  }
}

\node at ($(n30)!0.5!(n33)+(1.0,0)$) {$\cdots$}; 
\node at ($(n03)!0.5!(n33)+(0,1.0)$) {$\vdots$}; 
\node at ($(n33)+(0.5,0.5)$)        {$\cdot$};
\node at ($(n33)+(0.6,0.6)$)        {$\cdot$};
\node at ($(n33)+(0.7,0.7)$)        {$\cdot$};
\end{tikzpicture}
\caption{Infinite DAG corresponding to the 2D NOST channel in Example~\ref{ex:2DNOST-1}. Nodes in the lower and left boundaries are initial nodes; all other nodes are communication nodes.}
\end{figure}
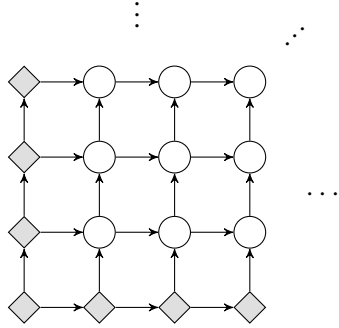

\begin{example}[Infinite binary tree] \label{ex:tree}
    Consider the
    poset-Markov channel
    with its DAG being the rooted infinite binary tree with all edges oriented away from the root.
    Formally, its set of nodes $\ms V$ is the set of all finite binary strings (denoted $\{0,1\}^*$) with the empty string 
    $\epsilon$ being the root, and its set of edges is
    \begin{equation}
        \ms E = \left\{ \, (w,w0) , \, (w,w1) \, \mid \, w \in \{0,1\}^* \right\} .
    \end{equation}
    The induced partial order on $\ms V$ posits $w_1 \prec w_2$ if and only if $w_1$ is a prefix of $w_2$.
    The set of\linebreak initial and communication nodes are $\ms I = \{\epsilon\}$ and $\ms C = \{0,1\}^* \setminus \{\epsilon\}$, respectively.\linebreak
    Each $\mc V_n$ is given by ${\mc V_n = \{ w \in \{0,1\}^* \ \big| \ |w| \le n \}}$.
    Markovianity of the channel constants is\linebreak ${P_{Y_w|X_{\mathrm{Pref(w)}},Y_{\mathrm{Pref(w)}}} = P_{Y_w|X_{w^-},Y_{w^-}}}$, $w \in \{0,1\}^* \setminus \{\epsilon\}$, where $\mathrm{Pref(w)}$ denotes the set of proper prefixes of $w$, and $w^-$ denotes the string $w$ with its final symbol deleted.
\end{example}

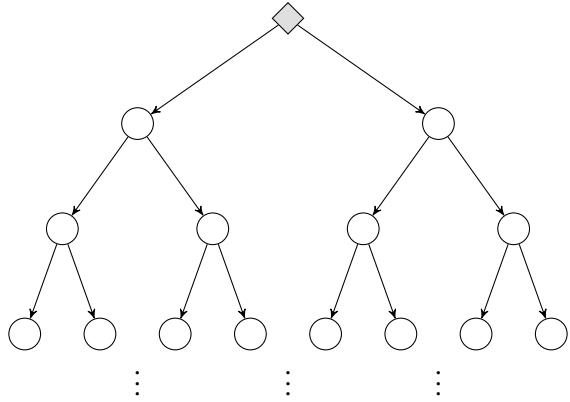
\begin{figure}[ht]
\centering
\begin{tikzpicture}[
  >=stealth',
  level distance=1.4cm,
  level 1/.style={sibling distance=4cm},
  level 2/.style={sibling distance=2cm},
  level 3/.style={sibling distance=1cm},
  edge from parent/.style={->,draw},     
  init/.style={diamond,draw,fill=gray!25,minimum size=12pt,inner sep=2pt},
  comm/.style={circle,draw,minimum size=12pt,inner sep=2pt}]

\node[init] (root) {} 
 child{ node[comm] {} 
   child{ node[comm] {} 
     child{ node[comm] {} } 
     child{ node[comm] {} }
   }
   child{ node[comm] {} 
     child{ node[comm] {} }
     child{ node[comm] {} }
   }
 }
 child{ node[comm] {} 
   child{ node[comm] {} 
     child{ node[comm] {} }
     child{ node[comm] {} }
   }
   child{ node[comm] {} 
     child{ node[comm] {} }
     child{ node[comm] {} }
   }
 };

\foreach \x in {-2.0,0,2.0}{
  \node at (\x,-4.75) {$\vdots$};
}
\end{tikzpicture}
\caption{Infinite binary tree in Example~\ref{ex:tree}.
The root is the initial node; all other nodes are communication nodes.}
\end{figure}



If
the poset-causal channel is poset-Markov, 
its feedback capacity defined in~\eqref{eq:poset-c-def} admits the following upper bound
where the mutual information terms and the variables reflect the Markov structure of the channel.

\begin{proposition}[Poset-Markov capacity upper bound] \label{prop:poset-relax}
    For a poset-Markov channel, its feedback capacity
    $C^{\textup{fb}}$
    is upper-bounded by
    the
    limit of convex optimizations 
    \begin{equation} \label{eq:poset-relax}
        \begin{aligned}
            C^{\textup{fb}} \ \le \ C^{\textup{fb}}_{\textup{u.b.}} \ \coloneq \
            \lim_{n\to\infty} & \max_{\left\{P_{X_v,Y_{\pa v}}\right\}_{v \in \mc C_n}} \frac{1}{|\mc C_n|} \sum_{v \in \mc C_n} I( X_v ; Y_v \mid Y_{\pa v} ) \\
            & \hspace{1cm} \textup{s.t.} \hspace{1.05cm} P_{Y_v} \hspace{0.4cm} = \sum_{x_v , y_{\pa v}} P_{Y_v \mid X_v , Y_{\pa v}} P_{X_v , Y_{\pa v}}, \\[-1ex]
            & \sum_{x_v , y_{\pa v \setminus \pa u}} P_{X_v , Y_{\pa v}} = \sum_{x_u , y_{\pa u \setminus \pa v}} P_{X_u , Y_{\pa u}} \quad \forall u,v \in \mc C_n ,
        \end{aligned}
    \end{equation}
    where the constraints impose channel constants as well as consistency at marginals
    that belong to more than one
    set of the form
    $\{v,\pa v\}$ where $v \in \mc C_n$.
\end{proposition}

\begin{proofsketch}
    We interpret the poset-Markov channel as an operationally equivalent 
    time-indexed channel,
    where the order of transmissions is according to the total order $\tau : \ms C \to \NN$,
    and appeal to~\eqref{eq:C-lim} for an initial characterization of feedback capacity.
    We then leverage Markovianity to upper-bound the directed information term by a sum of conditional mutual information terms,
    and also relax the variables to marginal joint distributions $\left\{P_{X_v,Y_{\pa v}}\right\}_{v \in \mc C_n}$ to obtain the right-hand side of~\eqref{eq:poset-relax}.
    The full proof can be found in
    Appendix~\ref{sec:proof-prop:poset-relax}.
\end{proofsketch}

In the next subsection,
we show that the
limit of the
convex problems
in~\eqref{eq:poset-relax}
equal a constant-sized optimization
under additional symmetry assumptions on the channel.


\subsection{Approximately Symmetric Channels}

In Proposition~\ref{prop:poset-relax}, we established an upper bound on the feedback capacity of
poset-Markov channels
in terms of the local distributions in the DAG.
We aim to obtain a \mbox{constant-sized} reduction of this upper bound
under vertex-transitivity so all vertices are ``the same''.
We obtain such a reduction for a subclass of poset-Markov channels that are induced subgraphs of \emph{Cayley graphs},
defined below.

\begin{definition}[Cayley graph~\cite{Cayley1878}]
    Let $\ms G$ be a group, and let $\mc S$ be a (finite) generating set of $\ms G$. The Cayley graph $\mathrm{Cay}(\ms G,\mc S) = (\ms V , \ms E)$ is a directed graph (i.e., $\ms E \subset \ms V \times \ms V$) with the following properties.
    \begin{enumerate}
        \item $\ms V = \ms G$, i.e., the vertices of $\mathrm{Cay}(\ms G,\mc S)$ are precisely the group elements $g \in \ms G$.
        \item Edges of $\mathrm{Cay}(\ms G,\mc S)$ are colored with $s$ colors: one color $c_s$ corresponding to each $s \in \mc S$.
        \item The set of edges is $\ms E = \left\{ (g,gs) \ \big| \ (g,s) \in \ms G \times \mc S \right\}$, i.e.,
        for every pair $(g,s) \in \ms G \times \mc S$, the ordered pair $(g,gs)$ is an edge in $\mathrm{Cay}(\ms G,\mc S)$. The color of the edge $(g,gs) \in \ms E$ is $c_s$.
    \end{enumerate}
\end{definition}

Through the symmetry properties of Cayley graphs, considering induced subgraphs of Cayley graphs for the poset-Markov channels facilitates leveraging isomorphic local neighborhoods of the nodes.
This in turn provides a means for single-letterization of the channel capacities of such poset-causal channels, which we call \emph{approximately symmetric channels}, defined as follows.

\begin{definition}[Approximately symmetric channel] \label{def:approx-sym}
    A \mbox{poset-Markov} channel on a DAG $(\ms V,\ms E)$ is called \emph{approximately symmetric} if it satisfies the following conditions.
        \begin{enumerate}[label={\arabic{enumi}-}, font=\textbf, align=left]
            \item[1-- (Symmetry via Cayley graph)] $(\ms V,\ms E)$ is an induced subgraph of an
            acyclic
            Cayley graph
            arising from a group with $d$ generators $s_1,\ldots,s_d$,\footnote{
            Such a Cayley graph is a DAG if no product $s_{i_1} \cdots s_{i_k}$ equals the identity element of the group for any $k$ and $i_1 , \ldots , i_k \in [d] $.
            }
            in a way that
            all communication nodes $v \in \ms C$ have in-degree $d_{\textup{in}} (v) = d$.
        
            \item[2-- (Vanishing fraction of boundary nodes)]
            For induced subgraphs $(\mc V',\mc E')$ of $(\ms V , \ms E)$, we say that a node $v \in \mc V'$ is a \emph{boundary node} of $\mc V'$ if it has in-degree or out-degree less than $d$ in $(\mc V',\mc E')$.
            Denote by $\mc B_n$ the set of boundary nodes of $\mc V_n$.
            Then, the fraction of boundary nodes $|\mc B_n| \big/ |\mc V_n| \to 0$ as $n\to\infty$.
            
            \item[3-- (Stationary channel constants)]
            There is a set of (universal) positive probability distributions $Q_{ Y | X , \overline{Y}' } \in \mathbb R_{>0}^{\mc Y \times \mc X \times \mc Y^d}$
            such that $P_{ Y_v \mid X_v , Y_{\pa v} } = Q_{ Y | X , \overline{Y}' }$ for all $v\in\ms C$,
            where each $Y_{\pa v}$ is the $d$-tuple $Y_{\pa v} = ( Y_{s_1^{-1} \cdot v} , \ldots , Y_{s_d^{-1} \cdot v} )$.
        \end{enumerate}
\end{definition}

\begin{figure}[ht]
\centering
\begin{tikzpicture}[scale=0.9,>=stealth']

\foreach \x in {0,...,5}{
  \foreach \y in {0,...,5}{
    \ifnum\x=0\relax
      \node[rectangle,draw] at (\x,\y) {};        
    \else\ifnum\x=5\relax
      \node[rectangle,draw] at (\x,\y) {};        
    \else\ifnum\y=0\relax
      \node[rectangle,draw] at (\x,\y) {};        
    \else\ifnum\y=5\relax
      \node[rectangle,draw] at (\x,\y) {};        
    \else
      \node[circle,fill=black!70,inner sep=2pt] at (\x,\y) {}; 
    \fi\fi\fi\fi
  }
}

\foreach \x in {0,...,4}{            
  \foreach \y in {0,...,5}{
    \draw[->] (\x+0.15,\y) -- (\x+0.85,\y);
  }
}

\foreach \x in {0,...,5}{
  \foreach \y in {0,...,4}{          
    \draw[->] (\x,\y+0.15) -- (\x,\y+0.85);
  }
}

\end{tikzpicture}
\caption{Depiction of vanishing fraction of boundary nodes in approximately symmetric channels.
Induced subgraph $6\times6$ square of the infinite 2D DAG.
Interior nodes (filled circles) vs. boundary nodes (hollow squares).}
\end{figure}
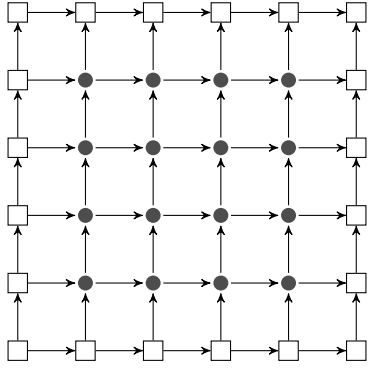

Below, we discuss examples (and nonexamples) of approximately symmetric poset-Markov channels.

\setcounter{example}{0}

\begin{example}[NOST channel--cont.] \label{ex:NOST-2}
    The NOST channel is a poset-Markov channel that is also approximately symmetric.
    Recall that it has nodes $\ms V = \ZZ_{\ge0}$ and edges $\ms E=\{(n,n+1) \mid n\in\mathbb{Z}_{\geq0}\}$. The Cayley graph (of which $(\ms V , \ms E)$ is an induced subgraph of) is
    $\mathrm{Cay}(\ZZ , \{1\})$, i.e.,
    the infinite directed line.
    The set of boundary nodes $\mc B_n$ of each $\mc V_n = \{0,\ldots,n\}$ is $\mc B_n = \{0,n\}$, which has vanishing fraction $2/(n+1) \to 0$.
    The channel constants are stationary with $P_{Y_n \mid X_n , Y_{n-1}} = Q_{Y \mid X , Y'},n\ge1$.
\end{example}


\setcounter{example}{3}

\begin{example}[2D NOST channel--cont.] \label{ex:2DNOST-2}
    The 2D NOST channel is approximately symmetric. 
    Recall that $\ms V = \ZZ^2_{\ge0}$, and $\ms E=\{\,(\,(i,j),(i,j+1)\,)\,,\,(\,(i,j),(i+1,j)\,) \mid i,j\ge0\,\}$.
    The Cayley graph in this case is
    $\mathrm{Cay}(\ZZ^2 , \{(0,1),(1,0)\})$, i.e., the Cayley graph of the group of pairs of integers with generators $(0,1),(1,0)$.
    The Cayley graph
    $\mathrm{Cay}(\ZZ^2 , \{(0,1),(1,0)\})$
    is the two-dimensional integer lattice where each node has outgoing edges upward and to the right.
    The set of boundary nodes $\mc B_n$ of each $\mc V_n = \{ (i,j) \mid 0 \le i,j \le n \}$ is $\mc B_n = \{ (i,j) \mid i \in \{0,n\} \ \textrm{or} \ j \in \{0,n\} \}$, which has vanishing fraction $4n/(n+1)^2 \to 0$.
\end{example}

\begin{nonexample}[Infinite binary tree--cont.] \label{ex:tree-2}
    The poset-Markov channel over the DAG of rooted infinite binary tree is \emph{not} approximately symmetric.
    Recall that $\ms V = \{0,1\}^*$ and
    \begin{equation}
        \ms E = \left\{ \, (w,w0) , \, (w,w1) \, \mid \, w \in \{0,1\}^* \right\} .
    \end{equation}
    In this case, the DAG $(\ms V , \ms E)$ is itself a Cayley graph with
    $(\ms V , \ms E) = \mathrm{Cay}(\langle0,1\rangle , \{0,1\})$ where $\langle0,1\rangle$ denotes the free group generated by $0$ and $1$.
    However, the set of boundary nodes $\mc B_n$ of each\linebreak ${\mc V_n = {\{ w \in \{0,1\}^* \ \big| \ |w| \le n \}}}$ is given by
    \begin{align}
        \mc B_n = \{\epsilon\} \cup \{ w \in \{0,1\}^* \ \big| \ |w| = n \} ,
    \end{align}
    and has fraction $(2^n + 1) / (2^{n+1} - 1) \to 1/2$, \emph{not} vanishing.
\end{nonexample}

Before we assert how the upper bound in Proposition~\ref{prop:poset-relax} single-letterizes, we provide a definition pertaining to equivalence of ordered subsets of $[d]$ that is needed to
capture the local consistency constraints in local convex relaxation in~\eqref{eq:poset-relax}.
For a node $v \in \ms C$ and an ordered subset $\mc S$ of $[d]$ with cardinality $|\mc S| = k$, denote by $\pai{\mc S}{v}$ the ordered $k$-tuple of the subset of $k$ parents of $v$ that are ordered according to the ordering in $\mc S$.

\begin{definition}[Equivalent subsets] \label{def:poset-equiv-ordered-subsets}
    Consider two ordered subsets $\mc S,\mc T$ of $[d]$.
    If there exist two nodes $u, v \in \ms C$ with $\pai{\mc S}{u} = \pai{\mc T}{v}$,
    then $\mc S \sim \mc T$.
\end{definition}

When the poset-Markov channel is further assumed to be approximately symmetric, the limit of maximizations in~\eqref{eq:poset-relax} converges to a single-letter problem, yielding the following single-letter upper bound on the feedback capacity.

\begin{theorem}[Single-letter upper bound] \label{thm:poset-single-letter}
    For an approximately symmetric poset-causal channel, the upper bound on its feedback capacity in Proposition~\ref{prop:poset-relax} single-letterizes as
    \begin{equation} \label{eq:poset-converse}
        \begin{aligned}
            C^{\textup{fb}} \ \le \ C^{\textup{fb}}_{\textup{u.b.}} \ = \
            & \max_{ P_{X , \ov Y'} } I(X;Y \mid \ov Y') \\[-2ex]
            & \ \textup{s.t.} \quad P_Y = P_{Y'_i} \ \ \forall i , \\[-2.25ex]
            & \hspace{0.86cm} P_{Y'_{\mc S}} = P_{Y'_{\mc T}} \ \ \forall \mc S \sim \mc T .
        \end{aligned}
    \end{equation}
    
\end{theorem}

\begin{proofsketch}
    Analogous to the proof of Proposition~\ref{prop:NOST-converse}, we make use of concavity together with Jensen's inequality to show that the average of the distributions $P_{X_v,Y_{\pa v}}, v \in \mc C_n$ has a larger objective function.
    Moreover, by the symmetry properties of the DAG, this average converges in the limit to a feasible point for the right-hand side.
    The full proof can be found in Appendix~\ref{sec:proof-thm:poset-single-letter}.
\end{proofsketch}



Finally, we instantiate Theorem~\ref{thm:poset-single-letter} for the NOST and the 2D~NOST channels.

\begin{corollary}[NOST] \label{corr:NOST}
    For the NOST channel, by Theorem~\ref{thm:poset-single-letter},
    \begin{equation} \label{eq:corr-NOST}
        C^{\textup{fb}}_{\textup{NOST}} \ \le \ \max_{ P_{X,Y'} } I( X ; Y \mid Y' ) \ \ \textup{s.t.} \ P_Y = P_{Y'} .
    \end{equation}
\end{corollary}
In Section~\ref{sec:NOST}, Theorem~\ref{thm:NOST} asserts that the upper bound in~\eqref{eq:corr-NOST} is tight,
i.e., the single-letter expression equals the NOST capacity.

\begin{corollary}[2D NOST] \label{corr:2DNOST}
    For the 2D NOST channel, Theorem~\ref{thm:poset-single-letter} yields the single-letter upper bound
    \begin{equation}
        C^{\textup{fb}}_{\textrm{2D}} \le \max_{P_{X,Y'_1,Y'_2}} I(X;Y \mid Y'_1,Y'_2) \ \ \textup{s.t.} \ P_Y = P_{Y'_1} = P_{Y'_2} ,
    \end{equation}
    where $Y'_1,Y'_2$ denote the outputs of the left neighbor and the bottom neighbor, respectively.
\end{corollary}

\section{Numerical Results} \label{sec:numerical}

In this section, we provide numerical results for the single-letter capacity upper bound given in Theorem~\ref{thm:poset-single-letter}.
We consider the binary NOST and 2D NOST channels for our simulations.

\subsection{NOST Channel}

The NOST channel has a particular ``neighboring dependency'' structure where the output of each transmission affects the next transmission.
The effect of this dependency structure on the capacity is manifested
in~\eqref{eq:corr-NOST}
through the constraint $P_Y = P_{Y'}$, which asserts that the maximal rate of communication is achieved through distributions that maintain \emph{output stationarity}.
If these neighboring dependencies are ignored, i.e., each transmission is isolated from the rest, the channel then would have the depiction in the figure below.
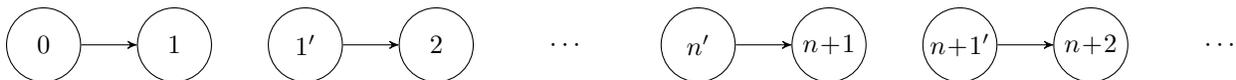
\begin{figure}[ht]
    \centering
    \begin{tikzpicture}[scale=1.45,>=stealth',
      every node/.style = {font=\small},
      init/.style={diamond,draw,fill=gray!25,minimum size=22pt,inner sep=1pt},
      comm/.style={circle,draw,minimum size=28pt,inner sep=0pt}]
      
        \node[comm] (n0) at (0,0) {$0$};
        \node[comm] (n1) at (1.2,0) {$1$};
        \node[comm] (n1') at (2.4,0) {$1'$};
        \node[comm] (n2) at (3.6,0) {$2$};

        \draw[->] (n0) -- (n1);
        \draw[->] (n1') -- (n2);
        
        \coordinate (cont) at ($(n2)+(1.2,0)$);
        \node at (cont) {$\cdots$};

        \node[comm] (n') at ($(cont)+(1.2,0)$) {$n'$};
        \node[comm] (n+1) at ($(cont)+(2.4,0)$) {$n\!+\!1$};
        \node[comm] (n+1') at ($(cont)+(3.6,0)$) {$n\!+\!1'$};
        \node[comm] (n+2) at ($(cont)+(4.8,0)$) {$n\!+\!2$};
        
        \draw[->] (n') -- (n+1);
        \draw[->] (n+1') -- (n+2);
        
        \coordinate (cont2) at ($(n+2)+(1.2,0)$);
        \node at (cont2) {$\cdots$};
    \end{tikzpicture}

    
    \caption{Depiction of the NOST channel with each transmission isolated.}
    \label{fig:NOST-myopic}
\end{figure}

\noindent Assuming the lack of the neighboring dependencies amounts to the
capacity upper bound
\begin{equation} \label{eq:NOST-myopic}
    C^{\textup{fb}}_{\textup{NOST}} \ \le \ \max_{ P_{X,Y'} } I( X ; Y \mid Y' ) . 
\end{equation}
This upper bound represents a \emph{myopic} strategy at a single node where full freedom on the feedback from the previous output $Y'$ is assumed,
without necessarily considering the effect of the output in the long run.
Denote the right-hand side of~\eqref{eq:NOST-myopic} by $C^{\text{myopic}}_{\textup{NOST}}$.
To observe the effect of the constraint $P_Y = P_{Y'}$,
we compare $C^{\textup{fb}}_{\textup{NOST}}$ and $C^{\text{myopic}}_{\textup{NOST}}$ for the particular NOST channel
below.

The classical binary Z-channel with error probability $\alpha$ has transition probabilities shown in the figure below.
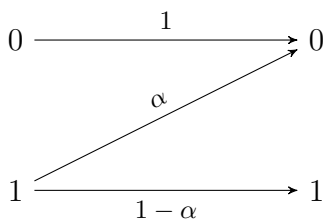
\begin{figure}[ht]
    \centering

    \begin{tikzpicture}[>=stealth',  
                    node font=\large,  
                    every node/.style={align=center},
                    x=4cm, y=2cm]      
        \node (L0) at (0,1) {$0$};
        \node (R0) at (1,1) {$0$};
        \node (L1) at (0,0) {$1$};
        \node (R1) at (1,0) {$1$};
        
        \draw[->] (L0) -- node[above,pos=0.5,font=\small] {$1$} (R0);
        
        \draw[->] (L1) -- node[below,pos=0.5,font=\small] {$1-\alpha$} (R1);
        
        \draw[->] (L1) -- node[pos=0.5,sloped,above,font=\small] {$\alpha$} (R0);
    \end{tikzpicture}


    \caption{Classical binary Z-channel with error probability $\alpha$.}
    \label{fig:z}
\end{figure}

\noindent We first define a
variable $M$ that is given by
\begin{equation}
    M =
    \begin{cases}
        X & \text{if } X = Y' , \\
        \mathrm{Bernoulli}(1/2) & \text{if } X \neq Y' .
    \end{cases}
\end{equation}
Then, $Y$ is the output of the Z-channel with error probability $\alpha$ and input $M$, where $\alpha$ is a noise parameter lying in $[0,1]$.
More explicitly, the channel is defined via the channel constants
\begin{equation} \label{eq:NOST-maj-vote-Q}
    Q_{Y \mid X , Y'} (0 \mid x , y') = 
    \begin{cases}
        1 & \text{if } (x,y') = (0,0) , \\[-1ex]
        \frac{1+\alpha}{2} & \text{if } (x,y') = (0,1) \textrm{ or } (1,0) , \\[-1ex]
        \alpha & \text{if } (x,y') = (1,1) .
    \end{cases}
\end{equation}
We call this channel the ``1D majority-vote Z-channel''.
The quantities
$C^{\textup{fb}}_{\textup{NOST}}$ and $C^{\text{myopic}}_{\textup{NOST}}$
for
this channel
are plotted in Figure~\ref{fig:NOST-Z-plot} below
as the noise parameter $\alpha$ varies in $[0,1]$.
\begin{figure}[ht]
    \centering
    \includegraphics[width=0.6\linewidth]{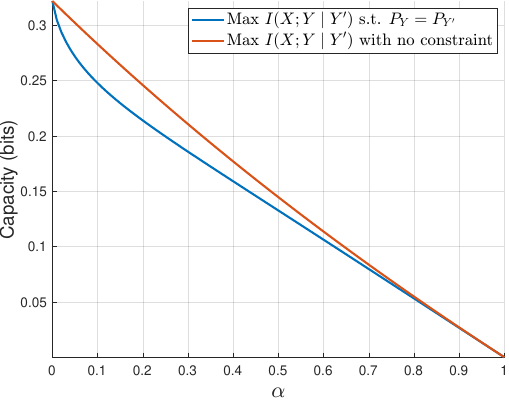}
    \caption{$C^{\textup{fb}}_{\textup{NOST}}$ and $C^{\text{myopic}}_{\textup{NOST}}$ for the NOST channel that is the 1D majority-vote Z-channel.}
    \label{fig:NOST-Z-plot}
\end{figure}

With no constraints, the distribution $P_{X,Y'}$ that maximizes $I(X;Y \mid Y')$ does not necessarily satisfy $P_Y = P_{Y'}$.
The constraint $P_Y = P_{Y'}$ thus results in a lower $I(X;Y \mid Y')$.
Consequently, the myopic bound $C^{\text{myopic}}_{\textup{NOST}}$ is a strict upper bound on the capacity $C^{\textup{fb}}_{\textup{NOST}}$ for $\alpha \in (0,1)$,
as apparent in Figure~\ref{fig:NOST-Z-plot}.
For $\alpha=1$, \eqref{eq:NOST-maj-vote-Q} implies that the output is $Y=1$ deterministically, so the capacity vanishes.

\subsection{2D NOST Channel}

Recall from Corollary~\ref{corr:2DNOST} that the single-letter upper bound on the feedback capacity of the 2D NOST channel---to be denoted by $C^{\text{bound}}_{\textrm{2D}}$---is given by
\begin{equation}
    C^{\text{bound}}_{\textrm{2D}} = \max_{P_{X,Y'_1,Y'_2}} I(X;Y \mid Y'_1,Y'_2) \ \ \textup{s.t.} \ P_Y = P_{Y'_1} = P_{Y'_2} .
\end{equation}
By relaxing the constraint $P_Y = P_{Y'_1} = P_{Y'_2}$, one obtains the ``myopic'' upper bound $C^{\text{myopic}}_{\text{2D}}$ on the feedback capacity, given by
\begin{equation}
    C^{\text{myopic}}_{\textrm{2D}} = \max_{P_{X,Y'_1,Y'_2}} I(X;Y \mid Y'_1,Y'_2) .
\end{equation}
We compare $C^{\text{bound}}_{\textrm{2D}}$ and $C^{\text{myopic}}_{\textrm{2D}}$ for two different instants of 2D NOST channels: the 2D majority-vote Z-channel that is analogous to
above, and the ``$Y^\prime$-asymmetric channel''.

\subsubsection{2D majority-vote Z-channel}

Let $M = \mathrm{maj}\{X,Y'_1,Y'_2\}$ denote the majority bit among $X,Y'_1,Y'_2$.
The output $Y$ is the result of $M$ passing through the classical binary Z-channel given in Figure~\ref{fig:z}.
This results in the channel constants given by
\begin{equation}
    Q_{Y \mid X , Y'_1 , Y'_2} (0 \mid x , y'_1 , y'_2) =
    \begin{cases}
        1 & \text{if } (x , y'_1 , y'_2) \in \{ (0,0,0) , (0,0,1) ,(0,1,0) ,(1,0,0) \} , \\[-1ex]
        \alpha & \text{if } (x , y'_1 , y'_2) \in \{ (1,1,1) , (1,1,0) ,(1,0,1) ,(0,1,1) \} .
    \end{cases}
\end{equation}
For the 2D NOST channel defined as such,
$C^{\text{bound}}_{\textrm{2D}}$ and $C^{\text{myopic}}_{\textrm{2D}}$ are plotted Figure~\ref{fig:2DNOST-Z-plot} below as the noise parameter $\alpha$ varies in $[0,1]$.
\begin{figure}[ht]
    \centering
    \includegraphics[width=0.54\linewidth]{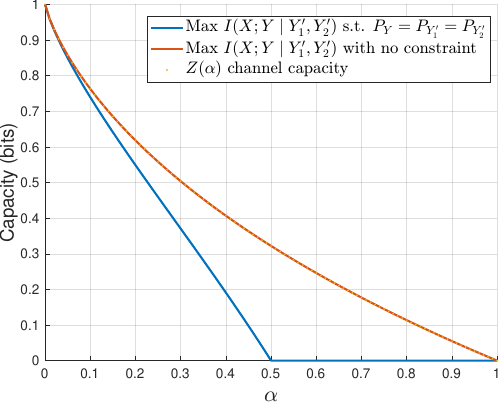}
    \caption{$C^{\textup{fb}}_{\text{2D}}$ and $C^{\text{myopic}}_{\text{2D}}$ for the 2D NOST channel that is the 2D majority-vote Z-channel.}
    \label{fig:2DNOST-Z-plot}
\end{figure}

\noindent When no constraint is present, one can set $(Y'_1,Y'_2)=(0,1)$ deterministically to get $M=X$.
This results in maximal $I(X;Y \mid Y'_1,Y'_2)$ and equals the capacity of the classical $Z$-channel, as seen in Figure~\ref{fig:2DNOST-Z-plot}.

Numerical observation suggests that
the optimal distribution $P_{X,Y'_1,Y'_2}$ does not satisfy $P_Y = P_{Y'_1} = P_{Y'_2}$,
hence the constraint $P_Y = P_{Y'_1} = P_{Y'_2}$ results in a lower $I(X;Y \mid Y'_1,Y'_2)$,
yielding a strictly better upper bound than that of the myopic strategy.

We observe that
when $\alpha \ge 0.5$,
the constraint $P_Y = P_{Y'_1} = P_{Y'_2}$ prevents any information flow, which results in $I(X;Y \mid Y'_1,Y'_2) = 0$.
The proof of this fact can be found in Appendix~\ref{sec:proof-alpha-ge-0.5}.

\subsubsection{\texorpdfstring{$Y^\prime$-asymmetric channel}{Y′-asymmetric channel}}

We demonstrate that the constraint $P_Y = P_{Y'_1} = P_{Y'_2}$ can be very restrictive when there is an asymmetry in the channel between the previous outputs $Y'_1$ and $Y'_2$.

Let the parameter $\alpha \in [0,1]$ represent the \emph{degree of asymmetry} between $Y'_1$, $Y'_2$.
We construct the channel such that $\alpha$ is the interpolation parameter between a $(Y'_1,Y'_2)$--symmetric and a \mbox{$(Y'_1,Y'_2)$--asymmetric} channel.
\begin{itemize}
    \item The \textbf{symmetric channel} is the noiseless majority-vote channel, i.e., $Y = \mathrm{maj}\{X,Y'_1,Y'_2\}$.
    \item The \textbf{asymmetric channel} is the channel defined by the logical operation $Y = (X \vee Y'_1) \wedge Y'_2$.
\end{itemize}
Notice that in the asymmetric channel $Y = (X \vee Y'_1) \wedge Y'_2$, for the maximal information flow through having $Y=X$, one would desire $Y'_1=0$ and $Y'_2=1$---hence the asymmetry.

Overall, we define the \emph{$Y^\prime$-asymmetric channel} as the interpolation of these two channels, i.e., 
\begin{equation}
    Y =
    \begin{cases}
      (X \vee Y'_1) \wedge Y'_2 & \text{with probability} \ \ \alpha , \\[-1ex]
      \mathrm{maj}\{X,Y'_1,Y'_2\} & \text{with probability} \ \ 1-\alpha .
    \end{cases}
\end{equation}
For the 2D NOST channel defined as such,
$C^{\text{bound}}_{\textrm{2D}}$ and $C^{\text{myopic}}_{\textrm{2D}}$ are plotted in Figure~\ref{fig:2DNOST-asym-plot} as the asymmetry parameter $\alpha$ varies in $[0,1]$.
\begin{figure}[ht]
    \centering
    \includegraphics[width=0.55\linewidth]{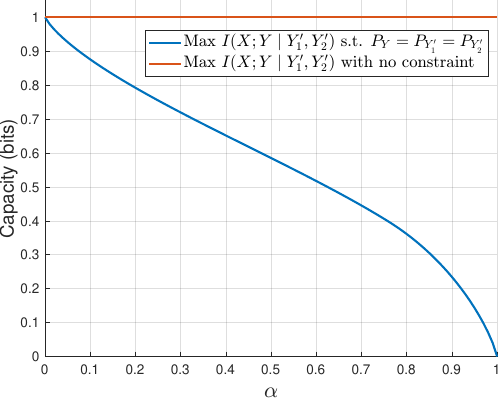}
    \caption{$C^{\textup{fb}}_{\text{2D}}$ and $C^{\text{myopic}}_{\text{2D}}$ for the 2D NOST channel that is the $Y^\prime$-asymmetric channel.}
    \label{fig:2DNOST-asym-plot}
\end{figure}

When there is no constraint, setting $(Y'_1,Y'_2) = (0,1)$ deterministically yields $Y=X$, hence\linebreak ${I(X;Y \mid Y'_1 , Y'_2) = 1}$.
But when the constraint $P_Y = P_{Y'_1} = P_{Y'_2}$ is present,
the maximal value of $I(X;Y \mid Y'_1 , Y'_2)$ gradually decreases to $0$
as the asymmetry increases with $\alpha$.

For $\alpha = 1$, the constraint $P_Y = P_{Y'_1} = P_{Y'_2}$ results in $I(X;Y \mid Y'_1 , Y'_2) = 0$. This is because $Y = (X \vee Y'_1) \wedge Y'_2$ implies
\begin{equation} \label{eq:y'-asym-proof-1}
    P_Y(1) = P_{X,Y'_1,Y'_2}(0,1,1) + P_{X,Y'_1,Y'_2}(1,0,1) + P_{X,Y'_1,Y'_2}(1,1,1) .
\end{equation}
But at the same time, due to $P_Y= P_{Y'_2}$,
\begin{equation} \label{eq:y'-asym-proof-2}
    P_Y(1) = P_{Y'_2}(1) = P_{X,Y'_1,Y'_2}(0,0,1) + P_{X,Y'_1,Y'_2}(0,1,1) + P_{X,Y'_1,Y'_2}(1,0,1) + P_{X,Y'_1,Y'_2}(1,1,1) .
\end{equation}
\eqref{eq:y'-asym-proof-1} and \eqref{eq:y'-asym-proof-2} imply $P_{X,Y'_1,Y'_2}(0,0,1) = 0$.
This means that either the event $(Y'_1,Y'_2) = (0,1)$ has probability $0$,
or given $(Y'_1,Y'_2) = (0,1)$, $X$ equals $1$ deterministically.
In either case, $Y = (X \vee Y'_1) \wedge Y'_2$ forces that given $Y'_1,Y'_2$, $Y$ is independent of $X$, hence $I(X;Y \mid Y'_1 , Y'_2) = 0$.

\section{Conclusion and Future Directions} \label{sec:conclusions}

In this paper, we study feedback capacities of communication channels whose inputs and outputs are indexed by elements of a poset.
For the subclass of channels satisfying
Markovianity and
approximate symmetry, we obtain an upper bound on the feedback capacity via a single-letter concave maximization in Theorem~\ref{thm:poset-single-letter}.
Our methodology for obtaining the single-letter upper bound is novel,
and relies on leveraging symmetry properties along with concavity.
Our framework recovers the single-letter expression for the capacity of the NOST channel with a novel proof methodology, and enables us to derive a \mbox{single-letter} upper bound for a new two-dimensional analogue of NOST.

The methodology we propose is capable of obtaining constant-sized bounds on capacity provided that certain Markovianity and approximate symmetry conditions hold solely in terms of the inputs and outputs of the channel.
As an example, our methodology does not apply to the trapdoor channel in Example~\ref{ex:trapdoor} due to the lack of such characteristics; more precisely, the channel lacking a simple description without the state variable.
It therefore would be of interest to understand whether our framework and methodology can be generalized to treat more general finite-state channels beyond the NOST channel.

Additional
future directions
include applying our methodology to modified variants of the setting,
such as
channels with continuous input and output alphabets,
or finite-state channels defined through a state variable that follows a hidden Markov model.
It would also be of interest to see whether our methodology can yield capacities of channels without feedback.



\newpage

\appendices

\section{Proof of Proposition~\ref{prop:NOST-formulation}} \label{sec:proof-prop:NOST-formulation}

\begin{silentproof}
    
    The NOST capacity is given by~\eqref{eq:C-lim}.
    We show that for each $n$ we have
    \begin{equation} \label{eq:NOST-lemma}
        \max_{\left\{P_{X_t|X^{t-1},Y^{t-1}}\right\}_{t=1}^n} I(X^n \to Y^n)
        \ = \
        \max_{ \left\{ P_{ X_t \mid Y_{t-1} } \right\}_{t=1}^n } \sum_{t=1}^n I( X_t ; Y_t \mid Y_{t-1} ) .
    \end{equation}
    We first prove $\text{LHS} \le \text{RHS}$, and then $\text{LHS} \ge \text{RHS}$.

    \noindent \underline{(i) $\text{LHS} \le \text{RHS}$:}
    The directed information on the LHS of \eqref{eq:NOST-lemma} is given by\linebreak ${I(X^n \to Y^n) = \sum_{t=1}^n I( X^t ; Y_t \mid Y^{t-1} )}$.
    For each of the rest of the summands, we have
    \begin{align}
        I( X^t ; Y_t \mid Y^{t-1} ) &= \underbrace{ H(Y_t \mid Y^{t-1}) }_{ \le H(Y_t \mid Y_{t-1}) } - \underbrace{ H( Y_t \mid X^t , Y^{t-1} ) }_{ = H( Y_t \mid X_t , Y_{t-1} ) } \\
        &\le H(Y_t \mid Y_{t-1}) - H( Y_t \mid X_t , Y_{t-1} )
        = I( X_t ; Y_t \mid Y_{t-1} ) .
    \end{align}
    Therefore, for their sum we get
    \begin{equation} \label{eq:NOST-mutual-info-decomp}
        \max_{\left\{P_{X_t|X^{t-1},Y^{t-1}}\right\}_{t=1}^n} I(X^n \to Y^n)
        \ \le \
        \max_{\left\{P_{X_t|X^{t-1},Y^{t-1}}\right\}_{t=1}^n} \sum_{t=1}^n I( X_t ; Y_t \mid Y_{t-1} ) .
    \end{equation}
    We will prove that the RHS of~\eqref{eq:NOST-mutual-info-decomp} is less than or equal to the RHS of~\eqref{eq:NOST-lemma}.
    For that, fix a set of variables $\left\{P'_{X_t|X^{t-1},Y^{t-1}}\right\}_{t=1}^n$ for the RHS of~\eqref{eq:NOST-mutual-info-decomp}.
    As in~\eqref{eq:joint-induced},
    these variables, together with the channel constants, induce a full-joint distribution $P'_{X^n,Y^n}$ through
    \begin{equation}
        P'_{X^n,Y^n} = \prod_{t=1}^n P'_{X_t|X^{t-1},Y^{t-1}} P_{Y_t|X^{t},Y^{t-1}}
        = \prod_{t=1}^n P'_{X_t|X^{t-1},Y^{t-1}} Q_{Y_t \mid X_t , Y_{t-1}} .
    \end{equation}
    Now construct the set of variables $\left\{ P_{ X_t \mid Y_{t-1} } \right\}_{t=1}^n$ for the RHS of~\eqref{eq:NOST-lemma} through
    \begin{align}
        P_{X_t \mid Y_{t-1}} &\coloneq P'_{X_t \mid Y_{t-1}}
        = \dfrac{ P'_{X_t , Y_{t-1}} }{ P'_{Y_{t-1}} } , \quad \ t=1,\ldots,n, \label{eq:NOST-defn-P-prime}
    \end{align}
    where both the numerator and the denominator in~\eqref{eq:NOST-defn-P-prime} are obtained by marginalizing the full-joint $P'_{X^n,Y^n}$.
    We show via induction that the joint distributions of $(X_t,Y_{t-1})$ for $t=1,\ldots,n$ are the same for the two sets of variables, i.e., $P_{X_t , Y_{t-1}} = P'_{X_t , Y_{t-1}}$ for $t=1,\ldots,n$.
    The base case holds as $P_{X_1} \coloneq P'_{X_1}$.
    For the inductive step, assume $P_{X_t , Y_{t-1}} = P'_{X_t , Y_{t-1}}$ for some $t$.
    We have that $P_{X_{t+1} , Y_t}$ equals
    \begin{align}
        P_{X_{t+1} , Y_t} (x_{t+1} , y_t) &= P_{X_{t+1} \mid Y_t} (x_{t+1} \mid y_t) \cdot P_{Y_t} (y_t) \label{eq:NOST-joint-equal-0} \\
        & \hspace{-1.5cm} = P_{X_{t+1} \mid Y_t} (x_{t+1} \mid y_t) \sum_{ x_t , y_{t-1} } Q_{ Y_t \mid X_t , Y_{t-1} } ( y_t \mid x_t , y_{t-1} ) \cdot P_{ X_t \mid Y_{t-1} } ( x_t \mid y_{t-1} ) \cdot P_{Y_{t-1}} (y_{t-1}) \label{eq:NOST-joint-equal-1} \\
        & \hspace{-1.5cm} = P'_{X_{t+1} \mid Y_t} (x_{t+1} \mid y_t) \sum_{ x_t , y_{t-1} } Q_{ Y_t \mid X_t , Y_{t-1} } ( y_t \mid x_t , y_{t-1} ) \cdot P'_{ X_t \mid Y_{t-1} } ( x_t \mid y_{t-1} ) \cdot P'_{Y_{t-1}} (y_{t-1}) \label{eq:NOST-joint-equal-2} \\
        & \hspace{-1.5cm} = P'_{X_{t+1} \mid Y_t} (x_{t+1} \mid y_t) \cdot P'_{Y_t} (y_t) 
        = P'_{X_{t+1} , Y_t} (x_{t+1} , y_t) \qquad \forall x_{t+1} \in \mc X , y_t \in \mc Y ,
    \end{align}
    where in~\eqref{eq:NOST-joint-equal-1} we used
    the channel output distributions in terms of the channel constants, and in~\eqref{eq:NOST-joint-equal-2} we used the construction of $\left\{P_{ X_t | Y_{t-1} }\right\}_{t=1}^n$ together with the inductive assumption.
    This concludes the induction and hence implies $P_{X_t , Y_{t-1}} = P'_{X_t , Y_{t-1}}$ for $t=1,\ldots,n$.
    By the channel constants, the same holds for the joint distributions of the triples $(Y_t , X_t , Y_{t-1})$ for $t=1,\ldots,n$, i.e.,
    \begin{align}
        P_{Y_t , X_t , Y_{t-1}} &= Q_{Y_t \mid X_t , Y_{t-1}} \cdot P_{X_t , Y_{t-1}} \\
        &= Q_{Y_t \mid X_t , Y_{t-1}} \cdot P'_{X_t , Y_{t-1}} = P'_{Y_t , X_t , Y_{t-1}} \label{eq:NOST-joint-equal-3}
    \end{align}
    and this in turn implies identical value of the conditional mutual information terms $I( X_t ; Y_t \mid Y_{t-1} )$ for $t=1,\ldots,n$.
    Therefore, for each feasible set of variables for the RHS of~\eqref{eq:NOST-mutual-info-decomp},
    there exists a feasible set of variables for the RHS of~\eqref{eq:NOST-lemma}
    with equal objective value.
    Hence, LHS $\le$ RHS in~\eqref{eq:NOST-lemma} follows.

    \noindent \underline{(ii) $\text{LHS} \ge \text{RHS}$:}
    We can lower bound the maximization on the LHS of \eqref{eq:NOST-lemma} by imposing constraints as
    \begin{equation} \label{eq:NOST-impose}
        \begin{aligned}
            \max_{\left\{P_{X_t|X^{t-1},Y^{t-1}}\right\}_{t=1}^n} I(X^n \to Y^n)
            \ \ge \
            &\max_{\left\{P_{X_t|X^{t-1},Y^{t-1}}\right\}_{t=1}^n} I(X^n \to Y^n) \\
            & \hspace{1.2cm} \textup{s.t. } \ P_{X_t|X^{t-1},Y^{t-1}} = P_{X_t \mid Y_{t-1}}\, , \ t=1,\ldots,n .
        \end{aligned}
    \end{equation}
    With the new constraints imposed, we have
    \begin{align}
        P_{Y_t \mid Y^{t-1}} (y_t \mid y^{t-1}) &= \sum_{x^t} P_{ Y_t \mid X^t , Y^{t-1} } ( y_t \mid x^t , y^{t-1} ) \cdot P_{ X^t \mid Y^{t-1} } ( x^t \mid y^{t-1} ) \notag \\
        &= \sum_{x^t} \underbrace{ P_{ Y_t \mid X^t , Y^{t-1} } ( y_t \mid x^t , y^{t-1} ) }_{ = P_{ Y_t \mid X_t , Y_{t-1} } ( y_t \mid x_t , y_{t-1} ) } \cdot \underbrace{ P_{ X_t \mid X^{t-1} , Y^{t-1} } ( x_t \mid x^{t-1} , y^{t-1} ) }_{ = P_{ X_t \mid Y_{t-1} } ( x_t \mid y_{t-1} ) } \cdot P_{ X^{t-1} \mid Y^{t-1} } ( x^{t-1} \mid y^{t-1} ) \notag \\
        &= \sum_{x_t} P_{ Y_t \mid X_t , Y_{t-1} } ( y_t \mid x_t , y_{t-1} ) \cdot P_{ X_t \mid Y_{t-1} } ( x_t \mid y_{t-1} ) \underbrace{ \sum_{x^{t-1}}  P_{ X^{t-1} \mid Y^{t-1} } ( x^{t-1} \mid y^{t-1} ) }_{=1} \notag \\
        &= P_{Y_t \mid Y_{t-1}} (y_t \mid y_{t-1}) .
    \end{align}
    For the directed information, this gives
    \begin{align}
        I(X^n \to Y^n) &= \sum_{t=1}^n I( X^t ; Y_t \mid Y^{t-1} )
        = \sum_{t=1}^n \underbrace{ H(Y_t \mid Y^{t-1}) }_{ = H(Y_t \mid Y_{t-1}) } - \underbrace{ H( Y_t \mid X^t , Y^{t-1} ) }_{ = H( Y_t \mid X_t , Y_{t-1} ) } \\
        &= \sum_{t=1}^n H(Y_t \mid Y_{t-1}) - H( Y_t \mid X_t , Y_{t-1} )
        = \sum_{t=1}^n I( X_t ; Y_t \mid Y_{t-1} ) .
    \end{align}
    As explained in the inductive steps of the LHS $\le$ RHS part above in \eqref{eq:NOST-joint-equal-0}--\eqref{eq:NOST-joint-equal-3},
    the sum\linebreak ${\sum_{t=1}^n I( X_t ; Y_t \mid Y_{t-1} )}$ is induced by the variables $\left\{ P_{ X_t \mid Y_{t-1} } \right\}_{t=1}^n$, hence
    \begin{equation}
        \begin{aligned}
            &\max_{\left\{P_{X_t|X^{t-1},Y^{t-1}}\right\}_{t=1}^n} I(X^n \to Y^n)
            \qquad = \qquad
            \max_{\left\{ P_{ X_t \mid Y_{t-1} } \right\}_{t=1}^n } \sum_{t=1}^n I( X_t ; Y_t \mid Y_{t-1} ) . \\[-0.5em]
            & \hspace{1.1cm} \textup{s.t. } \ P_{X_t|X^{t-1},Y^{t-1}} = P_{X_t \mid Y_{t-1}}\, , \\[-1.0em]
            & \hspace{2.25cm} t=1,\ldots,n
        \end{aligned}
    \end{equation}
    This, together with \eqref{eq:NOST-impose}, yields
    \begin{equation}
        \begin{aligned}
            \max_{\left\{P_{X_t|X^{t-1},Y^{t-1}}\right\}_{t=1}^n} I(X^n \to Y^n)
            \ \ge \
            \max_{ \left\{ P_{ X_t \mid Y_{t-1} } \right\}_{t=1}^n } \sum_{t=1}^n I( X_t ; Y_t \mid Y_{t-1} ) ,
        \end{aligned}
    \end{equation}
    and hence LHS $\ge$ RHS in~\eqref{eq:NOST-lemma}.

    As a result, LHS $\le$ RHS and LHS $\ge$ RHS imply LHS $=$ RHS in~\eqref{eq:NOST-lemma}, as desired.
    Normalizing and
    taking $n\to\infty$ in~\eqref{eq:NOST-lemma} yields Proposition~\ref{prop:NOST-formulation}.
\end{silentproof}

\section{Proof of Lemma~\ref{lem:NOST-indep-init-dist}} \label{sec:proof-lem:NOST-indep-init-dist}

\begin{silentproof}
    
    Denote by $\Delta = \Delta^{\mc Y}$ the set of probability distributions on the output alphabet $\mc Y$. Likewise, denote by $(\Delta_X)^Y = (\Delta^\mc X)^\mc Y = \left( P_{X \mid Y'}(\ \cdot \ \mid Y' = y') \in \Delta^\mc X \right)_{y' \in \mc Y}$ the set of all possible input distributions $P_{X \mid Y'}$, conditional on the previous output.

    Our goal is to use Banach fixed point theorem to show that for any $P_0$, the set of all feasible $P_{Y_n}$'s,
    as the input distributions range over all possible values,
    converges to a fixed subset $K \subset \Delta$ as $n \to \infty$, called the \emph{attracting fixed point}.

    For each fixed $P_{X \mid Y'}$, define $F_{P_{X \mid Y'}}$ sending $P_{Y'}$ to $P_Y$ given by:
    \begin{equation} \label{NOST-F-linear}
        P_{Y}(\cdot) = \sum_{(x,y') \in \mc X \times \mc Y} Q_{Y \mid X,Y'}(\cdot \mid x,y') \cdot P_{X \mid Y'}(x \mid y') \cdot P_{Y'}(y') .
    \end{equation}

    Endow $\Delta$ with the total variation metric $d_{\textup{TV}}(\cdot,\cdot)$, defined as
    \begin{equation}
        d_{\textup{TV}} \left( \, (p_y)_{y\in\mc Y} , (q_y)_{y\in\mc Y} \, \right) \coloneq \frac{1}{2} \sum_{y \in \mc Y} |p_y - q_y| .
    \end{equation}
    Since $d_{\textup{TV}}$ is equivalent to the $\ell_1$ metric, the space $(\Delta , d_{\textup{TV}})$ is compact with diameter
    \begin{equation}
        \mathrm{diam}_{\textup{TV}} (\Delta) \coloneq \max_{p,q \in \Delta} d_{\textup{TV}} (p,q) = 1.
    \end{equation}
    We show that all maps $F_{ P_{X|Y'} }$ are contractions in $d_{\textup{TV}}$ with Lipschitz constant
    \begin{equation} \label{eq:alpha}
        \alpha \coloneq 1 - |\mc Y| \cdot \gamma
        < 1 ,
    \end{equation}
    where $\gamma > 0$ is the minimal channel constant,
    given by
    \begin{equation}
        \gamma \coloneq \min_{y,x,y'} Q_{Y \mid X , Y'} (y \mid x , y') > 0.
    \end{equation}

    \begin{lemma} \label{lem:NOST-F-contr}
        For all $P_{X \mid Y'} \in (\Delta_X)^Y$, the map $F_{P_{X|Y'}}$ is a contraction with Lipschitz constant $\alpha$, i.e.,
        \begin{equation}
            d_{\textup{TV}} \left( F_{P_{X|Y'}}(P_Y) \, , \, F_{P_{X|Y'}}(R_Y) \right) \le \alpha \cdot d_{\textup{TV}} (P_Y , R_Y) \quad \forall P_{X \mid Y'} \in (\Delta_X)^Y , \ \forall P_Y , R_Y \in \Delta .
        \end{equation}
    \end{lemma}

    \begin{proof}
        Fix a $P_{X|Y'} \in (\Delta_X)^Y$, and consider the square matrix $[A_{y,y'}]_{y,y'\in\mc Y}$ given by
        \begin{equation} \label{eq:NOST-matrix-A}
            A = \left[ \sum\limits_{x\in\mc X} Q_{Y \mid X,Y'}(y \mid x,y') \cdot P_{X \mid Y'}(x \mid y') \right]_{y,y'\in\mc Y} .
        \end{equation}
        Notice that $A$ is the matrix that maps (in vector form) each $P_{Y'}$ to the corresponding ${P_Y = F_{P_{X \mid Y'}}(P_{Y'})}$.
        By assumption, each entry of $A$ is $\ge \gamma$. This implies~\cite[Prop.~5]{Lalley_MC_LecNotes} that $A$ is a contraction with respect to $d_{\textup{TV}}$ with Lipschitz constant $\alpha = 1 - |\mc Y| \cdot \gamma$.
        This yields Lemma~\ref{lem:NOST-F-contr}.
    \end{proof}
    
    For a given $P_{Y'} \in \Delta$, we are interested in the values its image $P_{Y} = F_{P_{X|Y'}} ( P_{Y'} )$ can take, as $P_{X|Y'}$ ranges over all of $(\Delta_X)^Y$. To that end, we extend $F$ to the set-valued map $\mc F : \Delta \to 2^\Delta$ where the image is the set of all $F_{P_{X|Y'}} ( P_{Y'} )$ as $P_{X|Y'}$ ranges over all of $(\Delta_X)^Y$, i.e.,
    \begin{equation} \label{eq:NOST-mc-F-defn}
        \mc F( P_{Y'} )
        := \left\{ F_{P_{X|Y'}} (P_{Y'} ) \ \left| \ P_{X|Y'} \in (\Delta_X)^Y \right. \right\} .
    \end{equation}

    We further extend $\mc F$ to map subsets of $\Delta$ instead of elements. Denote by $\mathrm{C}(\Delta)$ the set of closed subsets of $\Delta$. Recall that $\Delta$ is bounded in the metric $d_{\textup{TV}}$, so any closed subset of $\Delta$ is a closed and bounded subset.
    We extend $\mc F$ to the set-valued map $\overline{\mc F} : \mathrm{C}(\Delta) \to \mathrm{C}(\Delta)$, defined by
    \begin{equation} \label{eq:NOST-overline-F-defn}
        \overline{\mc F} (S) := \bigcup_{P_{Y'} \, \in \, S} \mc F( P_{Y'} ) \quad \forall S \in \mathrm{C}(\Delta) ,
    \end{equation}
    where the RHS is closed and bounded follows from~\cite[Lemma~3.11]{widder2009fixed}.
    The map $\overline{\mc F}$ enables us to represent the set of possible values of $P_{Y_t}$ as all the input distributions range over $(\Delta_X)^Y$. For an arbitrary $P_0 \in \Delta$, the set of possible $P_{Y_1}$'s is given by $\mc F( P_0 ) = \overline{\mc F} \left( \{ P_0 \} \right)$,
    and analogously, the set of possible $P_{Y_t}$'s is given by $\overline{\mc F}^t \left( \{ P_0 \} \right)$, for all $t\in\NN$.

    The spaces $2^{\Delta}$ and $\mathrm C(\Delta)$ are endowed with the Hausdorff distance $d_\textup{H}$ that is induced by the total variation metric $d_{\textup{TV}}$ on $\Delta$, as follows. Let $S_1,S_2 \in 2^{\Delta}$. For each $x\in S_1$, its distance to the set $S_2$ is given by $d_{\textup{TV}}(x,S_2) := \inf_{y\in S_2} d_{\textup{TV}}(x,y)$. Now denote
    \begin{equation}
        D(S_1||S_2) := \sup_{x\in S_1} d_{\textup{TV}}(x,S_2) .
    \end{equation}
    In words, $D(S_1||S_2)$ is the largest of the distances of the points in $S_1$ to the set $S_2$. The Hausdorff distance between $S_1$ and $S_2$ is then defined as the largest among $D(S_1||S_2)$ and $D(S_2||S_1)$, i.e.,
    \begin{align}
        d_{\textup{H}} (S_1 , S_2) &:= \max \left\{ D(S_1||S_2) \ , \ D(S_2||S_1) \right\} . 
    \end{align}

    Next, we prove that $\mc F$ is a contraction in $d_{\textup{H}}$ with the same Lipschitz constant $\alpha$ in Lemma~\ref{lem:NOST-F-contr}.

    \begin{lemma} \label{lem:NOST-mc-F-contraction}
        $\mc F$ is a contraction with the same Lipschitz constant $\alpha$ given in~\eqref{eq:alpha}, i.e., we have
        \begin{equation}
            d_{\textup{H}} \left( \, \mc F(P_Y) \, , \, \mc F(R_Y) \, \right) \le \alpha \cdot d_{\textup{TV}}(P_Y , R_Y) \quad \forall P_Y , R_Y \in \Delta .
        \end{equation}
    \end{lemma}

    \begin{proof}
        Consider arbitrary $P_Y , R_Y \in \Delta$,
        and fix any $T \in \mc F(P_Y)$. By~\eqref{eq:NOST-mc-F-defn}, we have $T = F_{P_{X|Y'}} (P_Y)$ for some $P_{X|Y'} \in (\Delta_X)^Y$.
        Note that for this particular $P_{X|Y'}$, we also have that $F_{P_{X|Y'}} (R_Y) \in \mc F(R_Y)$.
        This gives
        \begin{align}
            d_{\textup{TV}}( T , \mc F(R_Y) ) &:= \inf_{U\in\mc F(R_Y)} d_{\textup{TV}}(T,U) 
            \; = \inf_{U\in\mc F(R_Y)} d_{\textup{TV}} \left( F_{P_{X|Y'}} (P_Y) \, , \, U \right) \\
            &\; \le d_{\textup{TV}} \left( F_{P_{X|Y'}} (P_Y) \, , \, F_{P_{X|Y'}} (R_Y) \right) 
            \le \alpha \cdot d_{\textup{TV}}(P_Y , R_Y) .
        \end{align}
        Since $T \in \mc F(P_Y)$ was arbitrary, we conclude that
        \begin{align}
            D(\mc F(P_Y) \ || \ \mc F(R_Y)) &\coloneq \sup_{T\in\mc F(P_Y)} d_{\textup{TV}}( T , \mc F(R_Y) ) 
            \le \alpha \cdot d_{\textup{TV}}(P_Y , R_Y) .
        \end{align}
        Interchanging the roles of $P_Y$ and $R_Y$, we obtain
        $d_{\textup{H}} ( \mc F(P_Y) , \mc F(R_Y) ) \le \alpha \cdot d_{\textup{TV}}(P_Y , R_Y)$
        as desired, hence Lemma~\ref{lem:NOST-mc-F-contraction} follows.
    \end{proof}

    Next, we show that the set-valued extension $\overline{\mc F}$ of $\mc F$, defined in~\eqref{eq:NOST-overline-F-defn}, is also a contraction with the same Lipschitz constant $\alpha$.
    \begin{lemma} \label{lem:NOST-overline-F-contraction}
        $\overline{\mc F}$ is a contraction with the same Lipschitz constant $\alpha \in [0,1)$, i.e., we have
        \begin{equation}
            d_{\textup{H}} \left( \overline{\mc F} (S_1) , \overline{\mc F} (S_2) \right) \le \alpha \cdot d_{\textup{H}} (S_1 , S_2) \quad \forall S_1,S_2 \in \mathrm{C}(\Delta) .
        \end{equation}
    \end{lemma}

    \begin{proof}
        Take any $S_1,S_2 \in \mathrm{C}(\Delta)$. We first show that ${D \left( \overline{\mc F} (S_1) \ || \ \overline{\mc F} (S_2) \right) \le \alpha \cdot D(S_1 || S_2)}$, and then conclude by interchanging $S_1$ and $S_2$.
        We have
        \begin{align}
            D \left( \overline{\mc F} (S_1) \ || \ \overline{\mc F} (S_2) \right)
            &= \sup_{P_Y \in \overline{\mc F} (S_1)} d_{\textup{TV}}( P_Y , \overline{\mc F} (S_2) )
            && ( \text{definition of } D ) \\
            &= \sup_{P_Y \in \overline{\mc F} (S_1)} \inf_{R_Y \in \overline{\mc F} (S_2)} d_{\textup{TV}}(P_Y,R_Y) && (\text{definition of } d_{\textup{TV}} ) \\
            &= \sup_{P_Y \in \overline{\mc F} (S_1)} \inf_{R_{Y'} \in S_2} d_{\textup{TV}} \left( P_Y,\mc F(R_{Y'}) \right) && (\text{definition of } \overline{\mc F} ) \\
            &= \sup_{P_{Y'} \in S_1} \sup_{P_Y \in \mc F(P_{Y'})} \inf_{R_{Y'} \in S_2} d_{\textup{TV}} \left( P_Y,\mc F(R_{Y'}) \right) && (\text{definition of } \overline{\mc F} ) \\
            &\le \sup_{P_{Y'} \in S_1} \inf_{R_{Y'} \in S_2} \sup_{P_Y \in \mc F(P_{Y'})} d_{\textup{TV}} \left( P_Y,\mc F(R_{Y'}) \right) && (\text{max-min inequality}) \\
            &= \sup_{P_{Y'} \in S_1} \inf_{R_{Y'} \in S_2} D \left( \mc F(P_{Y'}) \ || \ \mc F(R_{Y'}) \right) && ( \text{definition of } D ) \\
            &\le \sup_{P_{Y'} \in S_1} \inf_{R_{Y'} \in S_2} d_{\textup{H}} \left( \mc F(P_{Y'}) \, , \,\mc F(R_{Y'}) \right) && ( \text{definition of } d_{\textup{H}} ) \\
            &\le \sup_{P_{Y'} \in S_1} \inf_{R_{Y'} \in S_2} \alpha \cdot d_{\textup{TV}} (P_{Y'},R_{Y'}) && (\text{by Lemma~\ref{lem:NOST-mc-F-contraction}}) \\
            &= \alpha \cdot \sup_{P_{Y'} \in S_1} d_{\textup{TV}} (P_{Y'},S_2)
            = \alpha \cdot D(S_1 || S_2) .
            && (\text{definitions of } d_{\textup{TV}} , D ) 
        \end{align}
        Interchanging $S_1$ and $S_2$, we obtain
        \begin{equation}
            \max \left\{ D \left( \overline{\mc F} (S_1) \ || \ \overline{\mc F} (S_2) \right) \, , \, D \left( \overline{\mc F} (S_2) \ || \ \overline{\mc F} (S_1) \right) \right\}
            \ \le \
            \alpha \cdot \max \left\{ D(S_1 || S_2) , D(S_2 || S_1) \right\} ,
        \end{equation}
        which is equivalent to
        $d_{\textup{H}} \left( \overline{\mc F} (S_1) , \overline{\mc F} (S_2) \right) \le \alpha \cdot d_{\textup{H}} (S_1 , S_2)$.
        This proves that $\overline{\mc F}$ is a contraction with the same Lipschitz constant $\alpha$ as well, hence Lemma~\ref{lem:NOST-overline-F-contraction} follows.
    \end{proof}

    Because $(\Delta,d_{\textup{TV}})$ is a complete metric space~\cite{kohlberg1982contraction}, we have that $(\mathrm{C}(\Delta) , d_{\textup{H}})$ is a complete metric space as well~\cite[Thm 3.5]{widder2009fixed}. This, together with the fact that the function $\overline{\mc F} : \mathrm{C}(\Delta) \to \mathrm{C}(\Delta)$ is a contraction, gives us the existence of a unique attracting fixed point of $\overline{\mc F}$, as given below.

    \begin{corollary} \label{cor:NOST-fixed-point}
        There is a unique closed subset $K \subset \Delta$ such that
        \begin{equation} \label{eq:NOST-closed-K}
            \lim_{n\to\infty} d_{\textup{H}} \left( \overline{\mc F}^n( \{P_0\} ) , K \right) = 0 \quad \forall P_0 \in \Delta .
        \end{equation}
        Moreover, for each $\epsilon>0$, there is an $N_\epsilon \in \NN$ such that we have
        \begin{equation} \label{eq:NOST-closed-K-2}
            d_{\textup{H}} \left( \overline{\mc F}^n( \{P_0\} ) , K \right) < \epsilon \quad \forall n \ge N_\epsilon \, , \, \forall P_0 \in \Delta .
        \end{equation}
    \end{corollary}

    \begin{proof}
        For~\eqref{eq:NOST-closed-K}, since $\overline{\mc F} : \mathrm{C}(\Delta) \to \mathrm{C}(\Delta)$ is a contraction on a complete metric space $(\mathrm{C}(\Delta) , d_{\textup{H}})$, by Banach fixed point theorem, it has a unique fixed point $K \in \mathrm{C}(\Delta)$ to which its iterates converge---in particular, starting from any singleton set $\{P_0\}$.

        For~\eqref{eq:NOST-closed-K-2}, the existence of such an $N_\epsilon$ for any $\epsilon>0$ can be shown as follows. Due to $\overline{\mc F}$ being a contraction with Lipschitz constant $\alpha < 1$, and $K$ being its fixed point, we have
        \begin{align}
            d_{\textup{H}} \left( \overline{\mc F}^n( \{P_0\} ) , K \right) &= d_{\textup{H}} \left( \overline{\mc F}^n ( \{P_0\} ) , \overline{\mc F}^n (K) \right) \\
            &\le \alpha^{n-1} \cdot d_{\textup{H}} \left( \overline{\mc F} ( \{P_0\} ) , \overline{\mc F} (K) \right)
            \le \alpha^{n-1} \cdot \mathrm{diam}(\Delta)
            = \alpha^{n-1}
        \end{align}
        where the last inequality follows by $\overline{\mc F} ( \{P_0\} ) , \overline{\mc F} (K) \in \Delta$. Now since $\alpha<1$, if we take $N_\epsilon$ large enough so that $\alpha^{N_\epsilon - 1} < \epsilon$, we have the desired result. Hence, Corollary~\ref{cor:NOST-fixed-point} follows.
    \end{proof}

    Next, we conclude the proof of~\eqref{eq:NOST-indep-init-dist}.
    Recall that the equality we want to prove is
    \begin{equation} \label{eq:NOST-indep-init-dist-copy}
        \begin{aligned}
            \lim_{n \rightarrow \infty} \frac{1}{n} & \max_{ \left\{ P_{ X_t \mid Y_{t-1} } \right\}_{t=1}^n } \sum_{t=1}^n I( X_t ; Y_t \mid Y_{t-1} )
            \ = \
            \lim_{n \rightarrow \infty} \frac{1}{n} \max_{ \left\{ P_{ X_t \mid Y_{t-1} } \right\}_{t=1}^n } \sum_{t=1}^n I( X_t ; Y_t \mid Y_{t-1} ) \\
            & \hspace{0.8cm} \textup{s.t.} \quad Y_0 \sim P_0
            \hspace{5.0cm} \textup{s.t.} \quad Y_0 \sim \widetilde{P}_0^{(n)} .
        \end{aligned}
    \end{equation}
    We prove a stronger fact where both the LHS and RHS have a sequence of distributions instead of the LHS having fixed distribution $P_0$.
    We show that for any two sequence $\{\widetilde R^{(n)}_{Y_0}\}_{n\ge1}$ and $\{\widetilde P^{(n)}_{Y_0}\}_{n\ge1}$ of initial distributions, we have
    \begin{equation} \label{eq:NOST-indep-init-dist-alternative}
        \begin{aligned}
            \lim_{n \rightarrow \infty} \frac{1}{n} & \max_{ \left\{ P_{ X_t \mid Y_{t-1} } \right\}_{t=1}^n } \sum_{t=1}^n I( X_t ; Y_t \mid Y_{t-1} )
            \ = \
            \lim_{n \rightarrow \infty} \frac{1}{n} \max_{ \left\{ P_{ X_t \mid Y_{t-1} } \right\}_{t=1}^n } \sum_{t=1}^n I( X_t ; Y_t \mid Y_{t-1} ) \\
            & \hspace{0.8cm} \textup{s.t.} \quad Y_0 \sim \widetilde{R}_0^{(n)}
            \hspace{4.7cm} \textup{s.t.} \quad Y_0 \sim \widetilde{P}_0^{(n)} ,
        \end{aligned}
    \end{equation}
    which will imply~\eqref{eq:NOST-indep-init-dist-copy} by taking $\{\widetilde R^{(n)}_{Y_0}\}_{n\ge1}$ to be the fixed sequence $\widetilde R^{(n)}_{Y_0} = P_0$.
    To avoid confusion, quantities $P^*, I^*$ will be referring to the LHS of~\eqref{eq:NOST-indep-init-dist-alternative}, whereas quantities $\widetilde P, \widetilde I$ will be referring to the RHS of~\eqref{eq:NOST-indep-init-dist-alternative}.

    We prove LHS $\ge$ RHS, and then conclude via symmetry.
    Fix an $\epsilon>0$, and let $\delta>0$ be such that the inequality $d_{\textup{TV}} ( P_{Y,X,Y'} , R_{Y,X,Y'} ) < \delta$ implies
    $\left| I^{(P)} (X;Y \mid Y') - I^{(R)} (X;Y \mid Y') \right| < \epsilon$, where $I^{(P)} (X;Y \mid Y')$ denotes the mutual information induced by the distribution $P_{Y,X,Y'}$, and analogously for the distribution $R$. The existence of such a $\delta$ follows by the continuity of mutual information
    which implies uniform continuity on the compact domain of $P_{Y,X,Y'}$.

    We define a set-valued function $M^{(n)} : \Delta \to \left( (\Delta_X)^Y \right)^n$ as
    \begin{equation}
        \begin{aligned}
            M^{(n)}(P'_{Y_0}) := &\argmax{ \left\{ P_{ X_t \mid Y_{t-1} } \right\}_{t=1}^n } \sum_{t=1}^n I( X_t ; Y_t \mid Y_{t-1} ) \\
            & \hspace{0.5cm} \textup{s.t.} \quad P_0 = P'_{Y_0}
        \end{aligned}
    \end{equation}
    that maps each $P'_{Y_0} \in \Delta$ to the corresponding set of maximizers $\{ P_{ X_t \mid Y_{t-1} } \}_{t=1}^n$.\footnote{
    If there are multiple maximizers, we arbitrarily choose one as the image under the function $M^{(n)}$.
    }

    As $\overline{\mc F}$ is a contraction with Lipschitz constant $\alpha < 1$,
    for each $\delta>0$, there is an $N_\delta \in \NN$
    with
    \begin{equation} \label{eq:NOST-triangle-ineq}
        d_{\textup{H}} \left( \overline{\mc F}^n( \{ \widetilde{R}_0^{(n)} \} ) , \overline{\mc F}^n( \{ \widetilde{P}_0^{(n)} \} ) \right) < \delta \quad \forall n \ge N_\delta .
    \end{equation}

    Now fix an $n > N_\delta$. Take
    \begin{equation} \label{eq:P-tilde-as-maximizer}
        ( \widetilde{P}_{ X_t \mid Y_{t-1} } )_{t=1}^n = M^{(n)} (\widetilde{P}_0^{(n)}) ,
    \end{equation}
    i.e., the variables $( \widetilde{P}_{ X_t \mid Y_{t-1} } )_{t=1}^n$ maximize $\sum_{t=1}^n I( X_t ; Y_t \mid Y_{t-1} )$ subject to $Y_0 \sim \widetilde{P}_0^{(n)}$.
    
    We construct
    a set of variables
    $\left\{ P^*_{ X_t \mid Y_{t-1} } \right\}_{t=1}^n$ as follows. Let $\left\{ P^*_{ X_t \mid Y_{t-1} } \right\}_{t=1}^{N_\delta}$ be such that the resulting $P^*_{Y_{N_\delta}}$
    satisfies $d_{\textup{TV}} \left( P^*_{Y_{N_\delta}} , \widetilde{P}_{Y_{N_\delta}} \right) < \delta$, which is possible due to~\eqref{eq:NOST-triangle-ineq}.
    Let the rest of the variables $\left\{ P^*_{ X_t \mid Y_{t-1} } \right\}_{t=N_\delta+1}^n$ be identical to $\left\{ \widetilde{P}_{ X_t \mid Y_{t-1} } \right\}_{t=N_\delta+1}^n$, i.e., $P^*_{ X_t \mid Y_{t-1} } = \widetilde{P}_{ X_t \mid Y_{t-1} }$ for $N_\delta+1 \le t \le n$.
    
    With this construction, we show by induction that the total variation distance between $P^*_{Y_t}$ and $\widetilde{P}_{Y_t}$ is less than $\delta$ for all $t \ge N_\delta$. The base case $t = N_\delta$ is true by the construction above. For the inductive step, if we assume that it holds for a certain $t \ge N_\delta$, i.e., that $d_{\textup{TV}} ( P^*_{Y_t} , \widetilde{P}_{Y_t} ) < \delta$, then we have
    \begin{align}
        d_{\textup{TV}} ( P_{Y_{t+1}} , \widetilde{P}_{Y_{t+1}} ) &= \frac{1}{2} \sum_{y\in\mc Y} \left| P_{Y_{t+1}} (y) - \widetilde{P}_{Y_{t+1}} (y) \right| \\
        &\le \frac{1}{2} \sum_{y\in\mc Y} \sum_{\substack{x_{t+1}\in\mc X \\ y_t\in\mc Y}} \hspace{-0.25cm} Q_{Y_{t+1} \mid X_{t+1},Y_t} (y \mid x_{t+1},y_t) \hspace{-0.075cm} \cdot \hspace{-0.075cm} P_{X_{t+1} \mid Y_t} (x_{t+1} \mid y_t) \hspace{-0.075cm} \cdot \hspace{-0.075cm} \left| P_{Y_{t}} (y_t) - \widetilde{P}_{Y_{t}} (y_t) \right| \\
        &= \frac{1}{2} \sum_{y_t\in\mc Y} \left| P_{Y_{t}} (y_t) - \widetilde{P}_{Y_{t}} (y_t) \right| 
        = d_{\textup{TV}} ( P_{Y_{t}} , \widetilde{P}_{Y_{t}} ) 
        < \delta ,
    \end{align}
    as desired. As a result of this induction, we have $d_{\textup{TV}} ( P^*_{Y_t} , \widetilde{P}_{Y_t} ) < \delta$ for all $t \ge N_\delta$.

    Next, we prove that the total variation distances between the joint distributions $P^*_{Y_t,X_t,Y_{t-1}}$ and $\widetilde{P}_{Y_t,X_t,Y_{t-1}}$ for $t > N_\delta$ are less than $\delta$ as well. Fix a $t \ge N_\delta + 1$. We know by the argument above that ${d_{\textup{TV}} ( P^*_{Y_t} , \widetilde{P}_{Y_t} ) < \delta}$. Using this, we have
    \begin{align}
        d_{\textup{TV}} & ( P_{Y_t,X_t,Y_{t-1}} , \widetilde{P}_{Y_t,X_t,Y_{t-1}} ) = \frac{1}{2} \sum_{\substack{ y_t\in\mc Y\\ x_t\in\mc X\\ y_{t-1}\in\mc Y }} \left| P_{Y_t,X_t,Y_{t-1}} (y_t,x_t,y_{t-1}) - \widetilde{P}_{Y_t,X_t,Y_{t-1}} (y_t,x_t,y_{t-1}) \right| \\
        &\le \frac{1}{2} \sum_{y_t\in\mc Y} \sum_{x_t\in\mc X} \sum_{y_{t-1}\in\mc Y} \hspace{-0.2cm} Q_{Y_{t} \mid X_{t},Y_{t-1}} (y_t \mid x_{t},y_{t-1}) \hspace{-0.075cm} \cdot \hspace{-0.075cm} P_{X_{t} \mid Y_{t-1}} (x_{t} \mid y_{t-1}) \hspace{-0.075cm} \cdot \hspace{-0.075cm} \left| P_{Y_{t-1}} (y_{t-1}) - \widetilde{P}_{Y_{t-1}} (y_{t-1}) \right| \\
        &= \frac{1}{2} \sum_{y_{t-1}\in\mc Y} \left| P_{Y_{t-1}} (y_{t-1}) - \widetilde{P}_{Y_{t-1}} (y_{t-1}) \right| \sum_{x_t\in\mc X} P_{X_{t} \mid Y_{t-1}} (x_{t} \mid y_{t-1}) \sum_{y_t\in\mc Y} Q_{Y_{t} \mid X_{t},Y_{t-1}} (y_t \mid x_{t},y_{t-1}) \\
        &= \frac{1}{2} \sum_{y_{t-1}\in\mc Y} \left| P_{Y_{t-1}} (y_{t-1}) - \widetilde{P}_{Y_{t-1}} (y_{t-1}) \right| 
        = d_{\textup{TV}} ( P_{Y_{t-1}} , \widetilde{P}_{Y_{t-1}} ) 
        < \delta ,
    \end{align}
    as desired.
    By the definition of $\delta$ in the beginning of the proof, this implies
    \begin{equation}
        \left| I^*(X_t;Y_t \mid Y_{t-1}) - \widetilde{I} (X_t;Y_t \mid Y_{t-1}) \right| < \epsilon \quad \forall t \ge N_\delta + 1 .
    \end{equation}
    
    As $n > N_\delta$ was arbitrary, we can let it grow arbitrarily large. Because the sums $\frac 1 n \sum_{t=1}^n I( X_t ; Y_t \mid Y_{t-1} )$ are averaged, the contribution of the first $N_\delta$ terms $\{ I (X_t;Y_t \mid Y_{t-1}) \}_{t=1}^{N_\delta}$ vanishes in the limit, hence we obtain
    \begin{equation}
        \left| \lim_{n \rightarrow \infty} \frac{1}{n} \sum_{t=1}^n I^*( X_t ; Y_t \mid Y_{t-1} )
        \ - \
        \lim_{n \rightarrow \infty} \frac{1}{n} \sum_{t=1}^n \widetilde{I} ( X_t ; Y_t \mid Y_{t-1} ) \right| < \epsilon
    \end{equation}
    which, in particular, implies
    \begin{equation} \label{eq:NOST-modified-sufficient-1}
        \lim_{n \rightarrow \infty} \frac{1}{n} \sum_{t=1}^n I^*( X_t ; Y_t \mid Y_{t-1} )
        \ > \
        \lim_{n \rightarrow \infty} \frac{1}{n} \sum_{t=1}^n \widetilde{I} ( X_t ; Y_t \mid Y_{t-1} ) - \epsilon.
    \end{equation}
    
    By the definition of
    $\{ \widetilde{P}_{ X_t \mid Y_{t-1} } \}_{t=1}^n$ in~\eqref{eq:P-tilde-as-maximizer},
    we have that $\sum_{t=1}^n \widetilde{I} ( X_t ; Y_t \mid Y_{t-1} )$ is equal to the supremum, i.e., we have
    \begin{equation} \label{eq:NOST-modified-sufficient-2}
        \begin{aligned}
            \sum_{t=1}^n \widetilde{I} ( X_t ; Y_t \mid Y_{t-1} )
            \ = \
            & \max_{ \left\{ P_{ X_t \mid Y_{t-1} } \right\}_{t=1}^n } \sum_{t=1}^n I( X_t ; Y_t \mid Y_{t-1} ) \\
            & \qquad \, \textup{s.t.} \quad Y_0 \sim \widetilde{P}_0^{(n)} .
        \end{aligned}
    \end{equation}
    
    However, the same is not true for $\sum_{t=1}^n I^*( X_t ; Y_t \mid Y_{t-1} )$ due to our specific construction of $\left\{ P^*_{ X_t \mid Y_{t-1} } \right\}_{t=1}^n$ rather than taking the maximizer. Thus, we instead have
    \begin{equation} \label{eq:NOST-modified-sufficient-3}
        \begin{aligned}
            \sum_{t=1}^n I^*( X_t ; Y_t \mid Y_{t-1} )
            \ \le \
            &\max_{ \left\{ P_{ X_t \mid Y_{t-1} } \right\}_{t=1}^n } \sum_{t=1}^n I( X_t ; Y_t \mid Y_{t-1} ) \\
            & \hspace{0.8cm} \textup{s.t.} \quad Y_0 \sim \widetilde{R}_0^{(n)} .
        \end{aligned}
    \end{equation}
    
    By~\eqref{eq:NOST-modified-sufficient-1},~\eqref{eq:NOST-modified-sufficient-2},~\eqref{eq:NOST-modified-sufficient-3}, we deduce
    \begin{equation}
        \begin{aligned}
            \lim_{n \rightarrow \infty} \frac{1}{n} & \max_{ \left\{ P_{ X_t \mid Y_{t-1} } \right\}_{t=1}^n } \sum_{t=1}^n I( X_t ; Y_t \mid Y_{t-1} )
            \ > \
            \lim_{n \rightarrow \infty} \frac{1}{n} \max_{ \left\{ P_{ X_t \mid Y_{t-1} } \right\}_{t=1}^n } \sum_{t=1}^n \widetilde{I}( X_t ; Y_t \mid Y_{t-1} ) - \epsilon \\
            & \hspace{0.8cm} \textup{s.t.} \quad Y_0 \sim \widetilde{R}_0^{(n)} \hspace{4.7cm} \textup{s.t.} \quad Y_0 \sim \widetilde{P}_0^{(n)} .
        \end{aligned}
    \end{equation}
    
    Since $\epsilon>0$ was arbitrary, we obtain
    \begin{equation}
        \begin{aligned}
            \lim_{n \rightarrow \infty} \frac{1}{n} & \max_{ \left\{ P_{ X_t \mid Y_{t-1} } \right\}_{t=1}^n } \sum_{t=1}^n I( X_t ; Y_t \mid Y_{t-1} )
            \ \ge \
            \lim_{n \rightarrow \infty} \frac{1}{n} \max_{ \left\{ P_{ X_t \mid Y_{t-1} } \right\}_{t=1}^n } \sum_{t=1}^n \widetilde{I} ( X_t ; Y_t \mid Y_{t-1} ) \\
            & \hspace{0.8cm} \textup{s.t.} \quad Y_0 \sim \widetilde{R}_0^{(n)} \hspace{4.7cm} \textup{s.t.} \quad Y_0 \sim \widetilde{P}_0^{(n)} .
        \end{aligned}
    \end{equation}
    which shows LHS $\ge$ RHS, as desired.
    By symmetry, LHS $\le$ RHS holds as well.
    Together, they imply~\eqref{eq:NOST-indep-init-dist-alternative},
    which in turn implies~\eqref{eq:NOST-indep-init-dist-copy}.
    This concludes the proof of Lemma~\ref{lem:NOST-indep-init-dist}.
\end{silentproof}

\section{Proof of Proposition~\ref{prop:NOST-achievability}} \label{sec:proof-prop:NOST-achievability}

\begin{silentproof}
    
    For any fixed $n$, we show
    \begin{equation} \label{eq:NOST-achievability-updated}
        \max_{\substack{ P_0 \\ \left\{P_{X_t \mid Y_{t-1}}\right\}_{t=1}^n }} \frac{1}{n} \sum_{t=1}^n I(X_t ; Y_t \mid Y_{t-1})
        \quad \ge \quad
        \max_{ P_{X,Y'} } I( X ; Y \mid Y' ) \quad \textup{s.t. } P_Y = P_{Y'} .
    \end{equation}
    We do this by showing that for any feasible variable for the RHS of~\eqref{eq:NOST-achievability-updated}, there exists feasible variables for the maximization on the LHS with
    objective value at least as large. To that end, fix a feasible $P'_{X , Y'}$ for the RHS of~\eqref{eq:NOST-achievability-updated}.
    For the LHS, set $P_0 \coloneq P'_{Y'}$ and $P_{X_t \mid Y_{t-1}} \coloneq P'_{X \mid Y'}$ for all $t$,
    With these choices, we prove $P_{X_t , Y_{t-1}} = P'_{X , Y'} \, \forall t$ through induction. For the base case $t=1$
    we have
    ${P_{X_1 , Y_0} = P_{X_1 \mid Y_0} \cdot P_0 = P'_{X \mid Y'} \cdot P'_{Y'} = P'_{X , Y'}}$,
    as desired.
    For the inductive step, assuming\linebreak ${P_{X_t , Y_{t-1}} = P'_{X , Y'}}$, first notice that this implies
    \begin{align}
        P_{Y_t} &= \sum_{x_t,y_{t-1}} Q_{Y_t \mid X_t , Y_{t-1}} \cdot P_{X_t , Y_{t-1}} = \sum_{x_t,y_{t-1}} Q_{Y \mid X , Y'} \cdot P'_{X , Y'} \\
        &= P'_{Y} = P'_{Y'} \hspace{4cm} \text{(by the constraint } P_Y = P_{Y'} \text{ in }~\eqref{eq:NOST-achievability-updated}\text{)}
    \end{align}
    which gives
    $P_{X_{t+1} , Y_t} = P_{X_{t+1} \mid Y_t} \cdot P_{Y_t} = P'_{X , Y'} \cdot P'_{Y'} = P'_{X , Y'}$,
    as desired, which concludes the inductive step. As a result, these variable choices
    yield $P_{X_t , Y_{t-1}} = P'_{X , Y'}$ for all $t \ge 1$. This implies
    \begin{align}
        P_{Y_t , X_t , Y_{t-1}} &= Q_{Y_t \mid X_t , Y_{t-1}} \cdot P_{X_t , Y_{t-1}} \\
        &= Q_{Y \mid X , Y'} \cdot P'_{X , Y'} 
        = P'_{Y, X , Y'} \quad \forall t\ge1 ,
    \end{align}
    which means that $I(Y_t \mid X_t , Y_{t-1}) = I( X ; Y \mid Y' )$ for all $t\ge1$.
    Hence, for the LHS of~\eqref{eq:NOST-achievability-updated}, for a fixed $n$, these choices of variables result in the objective value
    $I( X ; Y \mid Y' )$.
    As a result,
    Taking $n \to \infty$
    concludes the proof of Proposition~\ref{prop:NOST-achievability}.
\end{silentproof}


\section{Proof of Proposition~\ref{prop:NOST-converse}} \label{sec:proof-prop:NOST-converse}

\begin{silentproof}
    
    Define $D_\epsilon$ as the set of those $P_{X , Y'}$ satisfying the constraints of RHS \emph{$\epsilon$-approximately},~i.e.,
    \begin{equation} \label{eq:NOST-def-D-eps}
            D_{\epsilon} \coloneq \bigg\{ P_{X,Y'} \ : \ \left| P_{Y'}(y) - P_Y(y) \right| \le \epsilon \ \ \forall y \in \mc Y \, \bigg\} .
    \end{equation}

    The proof proceeds as follows.
    We show
    \begin{equation} \label{eq:NOST-instance-n}
        \max_{\substack{ P_0 \\ \left\{P_{X_t \mid Y_{t-1}}\right\}_{t=1}^n }} \frac{1}{n} \sum_{t=1}^n I(X_t ; Y_t \mid Y_{t-1})
        \quad \le \quad
        \max_{ P_{X , \ov Y'} \in D_{1/n} } I(X;Y \mid Y') .
    \end{equation}
    Afterwards, the proof will conclude by taking the limit $n\to\infty$ on both sides.

    To prove~\eqref{eq:NOST-instance-n}, we show that for any given feasible point for the maximization on the LHS, there is a feasible point for the maximization on the RHS that has objective value at least as large.
    To that end, take any feasible $P'_{Y_0} , \left\{P'_{X_t|Y_{t-1}}\right\}_{t=1}^n$ for the LHS. In what follows, the quantities associated with the symbol $'$ will indicate that it is induced by these variables.

    The variables $P'_{Y_0} , \left\{P'_{X_t|Y_{t-1}}\right\}_{t=1}^n$ induce joint distributions $\left\{P'_{X_t,Y_{t-1}}\right\}_{t=1}^n$ through the channel.
    Define the joint distribution $\widetilde P^{(n)}_{X , Y'}$ as the average of the joint distributions $\left\{P'_{X_t,Y_{t-1}}\right\}_{t=1}^n$, i.e.,
    \begin{equation} \label{eq:NOST-def-P-tilde}
        \widetilde P^{(n)}_{X , Y'} (x , y') \ \coloneq \ \frac{1}{n} \sum_{t=1}^n P'_{X_t,Y_{t-1}} (x , y') \quad \forall x,y' \in \mc X \times \mc Y .
    \end{equation}
    
    \begin{lemma} \label{lem:NOST-P-tilde-in-D}
        We have $\widetilde P^{(n)}_{X , Y'} \in D_{1/n}$.
    \end{lemma}


    \begin{proof}
        For simplicity,
        we drop the power ${(n)}$ and denote $\widetilde P^{(n)}$ by $\widetilde P$.
        The marginal $\widetilde P_{Y'} (y')$ equals
        \begin{align}
            \widetilde P_{Y'} (y') &= \sum_{x} \widetilde P_{X,Y'} (x,y')
            = \frac{1}{n} \sum_{t=1}^n \sum_{x} P'_{X_t,Y_{t-1}} (x,y') \\
            &= \frac{1}{n} \sum_{t=1}^n P'_{Y_{t-1}} (y')
            = \frac{1}{n} \sum_{t=0}^{n-1} P'_{Y_t} (y') \qquad \forall y' \in \mc Y . \label{eq:NOST-marginal-Y'}
        \end{align}
        Second, the marginal $\widetilde P_{Y} (y)$ can be written as
        \begin{align}
            \widetilde P_{Y} (y) &= \sum_{x,y'} Q_{Y \mid X,Y'} (y \mid x,y') \, \widetilde P_{X , Y'} (x , y') 
            = \sum_{x,y'} Q_{Y \mid X,Y'} (y \mid x,y') \frac{1}{n} \sum_{t=1}^n P'_{X_t,Y_{t-1}} (x,y') \\
            &= \frac{1}{n} \sum_{t=1}^n \sum_{x,y'} Q_{Y \mid X,Y'} (y \mid x,y') \cdot P'_{X_t,Y_{t-1}} (x,y')
            = \frac{1}{n} \sum_{t=1}^n P'_{Y_t} (y) \qquad \forall y \in \mc Y . \label{eq:NOST-marginal-Y}
        \end{align}
        By~\eqref{eq:NOST-marginal-Y'} and~\eqref{eq:NOST-marginal-Y}, the absolute difference of $\widetilde P_{Y'} (y)$ and $\widetilde P_{Y} (y)$ can be upper-bounded as
        \begin{align}
            \left| \widetilde P_{Y'}(y) - \widetilde P_Y(y) \right| &= \left| \frac{1}{n} \sum_{t=0}^{n-1} P'_{Y_t} (y) - \frac{1}{n} \sum_{t=1}^n P'_{Y_t} (y) \right| = \frac{1}{n} \underbrace{\left| P'_{Y_0} (y) - P'_{Y_n} (y) \right|}_{\le \, 1 - 0 \, = \, 1} 
            \le \frac{1}{n} \ \forall y \in \mc Y
        \end{align}
        which, by~\eqref{eq:NOST-def-D-eps}, implies $\widetilde P^{(n)}_{X , Y'} \in D_{1/n}$, as desired. This
        yields
        Lemma~\ref{lem:NOST-P-tilde-in-D}.
    \end{proof}

    The conditional mutual information $I(X;Y \mid Y')$ is a concave function of $P_{X , Y'}$, as the channel constants $Q_{Y \mid X , Y'}$ are fixed.
    Since $\widetilde P^{(n)}_{X , Y'}$ is an average, Jensen's inequality gives
    \begin{equation} \label{eq:NOST-jensen}
        \frac{1}{n} \sum_{t=1}^n I'( X_t ; Y_t \mid Y_{t-1} )
        \ \le \
        I(X;Y \mid Y') \, \Bigg|_{ P_{X , Y'} \, = \, \widetilde P^{(n)}_{X , Y'} }
    \end{equation}
    where the mutual information terms $I'$ on the LHS are induced by the variables $P'_{Y_0} , \left\{P'_{X_t|Y_{t-1}}\right\}_{t=1}^n$.
    Maximizing both the LHS and the RHS in~\eqref{eq:NOST-jensen}, we get
    \begin{equation} \label{eq:NOST-instance-n-2}
        \max_{\substack{ P_0 \\ \left\{P_{X_t \mid Y_{t-1}}\right\}_{t=1}^n }} \frac{1}{n} \sum_{t=1}^n I(X_t ; Y_t \mid Y_{t-1})
        \ \le \
        \max_{ P_{X , Y'} \in D_{1/n} } I(X;Y \mid Y') .
    \end{equation}
    Each $D_{1/n}$ is a compact set by the definition in~\eqref{eq:NOST-def-D-eps}, and the sequence of sets $\{D_{1/n}\}_{n\ge1}$ is therefore a decreasing sequence of compact sets,
    i.e.,
    $D_{1/(n+1)} \subset D_{1/n}$, and hence has nonempty limit $D_0$, given by
    $D_0 \coloneq \left\{ P_{X , Y'} \, \big| \, P_Y = P_{Y'} \right\}$,
    which is precisely the constraint set of the single-letter expression on the RHS of~\eqref{eq:NOST-converse}. Therefore, taking $n$ to infinity on both sides of~\eqref{eq:NOST-instance-n-2},
    \begin{align}
        \lim_{n\to\infty} \max_{\substack{ P_0 \\ \left\{P_{X_t \mid Y_{t-1}}\right\}_{t=1}^n }} \frac{1}{n} \sum_{t=1}^n I(X_t ; Y_t \mid Y_{t-1})
        \quad &\le \quad
        \lim_{n\to\infty} \max_{ P_{X , Y'} \in D_{1/n} } I(X;Y \mid Y') \\
        &= \quad \max_{ P_{X , Y'} \in D_{0} } I(X;Y \mid Y') \\
        &= \quad \max_{ P_{X,Y'} } I( X ; Y \mid Y' ) \quad \textup{s.t. } P_Y = P_{Y'}
    \end{align}
    as desired. This concludes the proof of Proposition~\ref{prop:NOST-converse}.
\end{silentproof}



\section{Proof of Proposition~\ref{prop:poset-relax}} \label{sec:proof-prop:poset-relax}

\begin{silentproof}

    The proof is based on the capacity characterization of discrete-time channels given in~\eqref{eq:C-lim}.
    Recall the equivalent interpretation of the poset-causal channel where the communication nodes transmit according to the total order $\tau : \ms C \to \NN$.
    By considering the subsequence $\{\mc C_n\}_{n\ge1}$ via~\eqref{eq:tau-subseq}, this interpretation
    gives the feedback capacity through~\eqref{eq:C-lim} as
    \begin{equation} \label{eq:poset-time-index}
        C^{\textup{fb}} (P_0) = \lim_{n\to\infty} \frac{1}{|\mc C_n|} \max_{ \left\{ P_{X_t|X^{t-1},Y^{t-1}} \right\}_{t=1}^{|\mc C_n|} } I(X^{|\mc C_n|}\to Y^{|\mc C_n|}) .
    \end{equation}
    We show that for each $n$, the maximization instance in~\eqref{eq:poset-time-index} is upper-bounded by the maximization instance on the RHS of~\eqref{eq:poset-relax}.
    Fix a set of feasible variables $\left\{ P'_{X_t|X^{t-1},Y^{t-1}} \right\}_{t=1}^{|\mc C_n|}$ for~\eqref{eq:poset-time-index}.
    Through the channel, this induces a full-joint distribution $P'_{X_{\mc C_n},Y_{\mc V_n}}$ given by the factorization
    \begin{equation}
        P'_{X_{\mc C_n},Y_{\mc V_n}} = P_0 \cdot \prod_{v \in \mc C_n} P'_{X_{\tau(v)} \mid X^{\tau(v)-1},Y^{\tau(v)-1}} \cdot \prod_{v \in \mc C_n} Q_{Y_v \mid X_v , Y_{\pa v}} .
    \end{equation}
    We set $P_{X_v,Y_{\pa v}} , v \in \mc C_n$ to be the corresponding marginalizations of the full-joint distribution $P'_{X_{\mc C_n},Y_{\mc V_n}}$, i.e.,
    \begin{equation} \label{eq:poset-p-p'}
        P_{X_v,Y_{\pa v}} \coloneq P'_{X_v,Y_{\pa v}} \quad \forall v \in \mc C_n ,
    \end{equation}
    which by construction satisfies the constraints on the RHS of~\eqref{eq:poset-relax}.
    
    On the other hand,
    the directed information on the RHS of~\eqref{eq:poset-time-index} is given by
    \begin{equation}
        I( X^{|\mc C_n|} \to Y^{|\mc C_n|} ) 
        = \sum_{t=1}^{|\mc C_n|} I( X^t ; Y_t \mid Y^{t-1} ) .
    \end{equation}
    Consider each summand $I( X^t ; Y_t \mid Y^{t-1} )$ for some time index $t$. Denoting $t = \tau(v)$, we have
    \begin{align}
        I( X^t ; Y_t \mid Y^{t-1} ) &= \underbrace{ H(Y_t \mid Y^{t-1}) }_{ \le H(Y_t \mid Y_{\pa v}) } - \underbrace{ H( Y_t \mid X^t , Y^{t-1} ) }_{ = H( Y_t \mid X_t , Y_{\pa v} ) } \\
        &\le H(Y_t \mid Y_{\pa v}) - H( Y_t \mid X_t , Y_{\pa v} ) \\
        &= I( X_t ; Y_t \mid Y_{\pa v} ) 
        = I( X_v ; Y_v \mid Y_{\pa v} ) .
    \end{align}
    Summing these up,
    we obtain
    \begin{align}
        I( X^{|\mc C_n|} \to Y^{|\mc C_n|} ) &\le \sum_{v \in \mc C_n} I( X_v ; Y_v \mid Y_{\pa v} ) 
        = \sum_{t=1}^{|\mc C_n|} I( X_{\tau^{-1}(t)} ; Y_{\tau^{-1}(t)} \mid Y_{ \pa{ \tau^{-1}(t) } } ) .
    \end{align}
    This implies the inequality
    \begin{equation} \label{eq:lower-bound-init}
        \max_{ \left\{ P_{X_t|X^{t-1},Y^{t-1}} \right\}_{t=1}^{|\mc C_n|} } I(X^{|\mc C_n|} \to Y^{|\mc C_n|})
        \le
        \max_{ \left\{ P_{X_t|X^{t-1},Y^{t-1}} \right\}_{t=1}^{|\mc C_n|} } \sum_{t=1}^{|\mc C_n|} I( X_{\tau^{-1}(t)} ; Y_{\tau^{-1}(t)} \mid Y_{ \pa{ \tau^{-1}(t) } } ) .
    \end{equation}
    For a fixed $t = \tau(v)$, each term $I( X_{\tau^{-1}(t)} ; Y_{\tau^{-1}(t)} \mid Y_{ \pa{ \tau^{-1}(t) } } ) = I( X_v ; Y_v \mid Y_{\pa v} )$ is a function of the distribution $P_{ X_v , Y_{\pa v} }$
    via the channel constant $Q_{ Y_v \mid X_v , Y_{\pa v} }$.
    Therefore, by the construction in~\eqref{eq:poset-p-p'}, the RHS of~\eqref{eq:lower-bound-init} has the same objective value as the sum $\sum_{v \in \mc C_n} I( X_v ; Y_v \mid Y_{\pa v} )$ induced by ${P_{X_v,Y_{\pa v}} , v \in \mc C_n}$ that are defined in~\eqref{eq:poset-p-p'}.
    Therefore, for each set of feasible variables for the RHS of~\eqref{eq:poset-time-index},
    there exists a set of feasible variables for the RHS of~\eqref{eq:poset-relax}.
    This yields Proposition~\ref{prop:poset-relax}.
\end{silentproof}

\section{Proof of Theorem~\ref{thm:poset-single-letter}} \label{sec:proof-thm:poset-single-letter}

\begin{silentproof}
    
    We want to show that for a poset-causal channel that is approximately symmetric, we have
    \begin{equation} \label{eq:poset-single-let-restate}
        \begin{aligned}
            C^{\textup{fb}}_{\textup{u.b.}} \quad \coloneq \quad
            \lim_{n\to\infty} & \max_{\left\{P_{X_v,Y_{\pa v}}\right\}_{v \in \mc C_n}} \frac{1}{|\mc C_n|} \sum_{v \in \mc C_n} I( X_v ; Y_v \mid Y_{\pa v} )
            \quad = \quad
            \max_{ P_{X , \ov Y'} } I(X;Y \mid \ov Y') \\
            & \hspace{-1cm} \qquad \textup{s.t.} \ \ P_{Y_v} = \sum_{x_v , y_{\pa v}} P_{Y_v \mid X_v , Y_{\pa v}} P_{X_v , Y_{\pa v}},
            \hspace{1.9cm} \textup{s.t.} \quad P_Y = P_{Y'_i} \ \forall i , \\[-1ex] 
            & \hspace{-2.425cm} \sum_{x_v , y_{\pa v \setminus \pa u}} P_{X_v , Y_{\pa v}} = \sum_{x_u , y_{\pa u \setminus \pa v}} P_{X_u , Y_{\pa u}} \ \forall u,v \in \mc C_n
            \hspace{2.125cm} P_{Y'_{\mc S}} = P_{Y'_{\mc T}} \ \forall \mc S \sim \mc T . 
        \end{aligned}
    \end{equation}
    Denote the consistency linear equality constraints on the LHS of~\eqref{eq:poset-single-let-restate} by $\mathrm{LE^{(n)}_{\text{c}}}$ for the $n$\textsuperscript{th} maximization instance.
    
    Define $D_\epsilon$ as the set of those $P_{X , \ov Y'}$ satisfying the constraints of RHS \emph{$\epsilon$-approximately},~i.e.,
    \begin{equation} \label{eq:poset-upper-b-def-D-eps}
        \begin{aligned}
            D_{\epsilon} \coloneq \bigg\{ P_{X , \overline{Y}'} \ : \ & \left| P_{Y'_i}(y) - P_Y(y) \right| \le \epsilon \ \ \forall y \in \mc Y , \, \forall i \in [d] , \\[-2ex]
            & \left| P_{Y'_{\mc S}}(\ov y) - P_{Y'_{\mc T}}(\ov y) \right| \le \epsilon \ \ \forall \ov y \in \mc Y^{|\mc S|} , \, \forall \mc S \sim \mc T \bigg\} .
        \end{aligned}
    \end{equation}
    Also denote $\delta_n \coloneq 4 \dfrac{|\mc B_n|}{|\mc V_n|}$, which by assumption converges to $0$ as $n$ grows.
    
    We show for each $n\ge1$ that
    \begin{equation} \label{eq:poset-upper-b-instance}
        \begin{aligned}
            &\max_{ \left\{ P_{X_v , Y_{\pa v}} \right\}_{v \in \mc C_n} } \ \frac{1}{|\mc C_n|} \sum_{v \in \mc C_n} I( X_v ; Y_v \mid Y_{\pa v} )
            \quad \le \quad
            \max_{ P_{X , \ov Y'} \in D_{\delta_n} } I(X;Y \mid \ov Y') . \\[-1.5ex]
            & \hspace{1.0cm} \textup{s.t.} \qquad \mathrm{LE^{(n)}_{\text{c}}}
        \end{aligned}
    \end{equation}
    Afterwards, we conclude the proof by taking $n\to\infty$.

    To prove~\eqref{eq:poset-upper-b-instance}, we show that for any given feasible point for the maximization on the LHS, there is a feasible point for the maximization on the RHS that has objective value at least as large.
    To that end, take any feasible $\left\{ P'_{X_v , Y_{\pa v}} \right\}_{v \in \mc C_n}$ for the LHS. In what follows, the quantities associated with the symbol $'$ will indicate that it is induced by these variables.

    Define $\widetilde P^{(n)}_{X , \overline{Y}'}$ as the average of $P'_{X_v , Y_{\pa v}}$ for $v \in \mc C_n$, i.e.,
    \begin{equation} \label{eq:poset-upper-b-def-P-tilde}
        \widetilde P^{(n)}_{X , \overline{Y}'} (x , \overline{y}') \ \coloneq \ \frac{1}{|\mc C_n|} \sum_{v \in \mc C_n} P'_{ X_v , Y_{\pa v} } (x , \overline{y}') \quad \forall x,\overline{y}' .
    \end{equation}


    We show that $\widetilde P^{(n)}_{X , \overline{Y}'} \in D_{\delta_n}$.
    We omit the power ${(n)}$ and denote $\widetilde P^{(n)}$ by $\widetilde P$ for simplicity.    

    We start with analyzing the marginals $\widetilde P_{Y'_i}$ for a fixed $i$, and then $\widetilde P_Y$.
    Denote ${\ov y' \coloneq (y'_1 , \ldots , y'_d)}$, and let $\overline{y}_{ \setminus \{i\} }'$ denote the ($d-1$)-tuple obtained by removing $y'_i$ from the $d$-tuple~$\overline{y}'$.
    $\widetilde P_{Y'_i} (y'_i)$ can be written as
    \begin{align}
        \widetilde P_{Y'_i} (y'_i) &= \sum_{x , \overline{y}_{ \setminus \{i\} }'} \hspace{-0.1cm} \widetilde P_{X , \overline{Y}'} (x , \overline{y}') = \frac{1}{|\mc C_n|} \sum_{v \in \mc C_n} \sum_{x , \overline{y}_{ \setminus \{i\} }'} \hspace{-0.1cm} P'_{ X_v , Y_{\pa v} } (x , \overline{y}') && \text{(marginalization, \eqref{eq:poset-upper-b-def-P-tilde})} \\
        &= \frac{1}{|\mc C_n|} \sum_{v \in \mc C_n} P'_{ Y_{s_i^{-1} \cdot v} } (y'_i) && \text{(marginalization)} \\
        &= \frac{1}{|\mc C_n|} \sum_{v \in s_i^{-1} \cdot \mc C_n} P'_{ Y_{v} } (y'_i) && \text{(change of variables)} \\
        &= \frac{1}{|\mc C_n|} \hspace{-0.1cm} \left( \sum_{v \, \in \, s_i^{-1} \cdot \mc C_n \, \setminus \, \mc C_n} \hspace{-0.7cm} P'_{ Y_{v} } (y'_i) + \hspace{-0.5cm} \sum_{v \, \in \, s_i^{-1} \cdot \mc C_n \, \cap \, \mc C_n} \hspace{-0.7cm} P'_{ Y_{v} } (y'_i) \right) \, \forall y'_i \in \mc Y . \label{eq:first-sum}
    \end{align}
    Analogously, $\widetilde P_Y(y)$ can be written as
    \begin{align}
        \widetilde P_Y(y) &= \sum_{x,\overline{y}'} Q_{Y \mid X,\overline{Y}'} (y \mid x,\overline{y}') \cdot \widetilde P_{X , \overline{Y}'} (x , \overline{y}') && \text{(channel output)} \\
        &= \sum_{x,\overline{y}'} Q_{Y \mid X,\overline{Y}'} (y \mid x,\overline{y}') \cdot \frac{1}{|\mc C_n|} \sum_{v \in \mc C_n} P'_{ X_v , Y_{\pa v} } (x , \overline{y}') && \text{(by~\eqref{eq:poset-upper-b-def-P-tilde})} \\
        &= \frac{1}{|\mc C_n|} \sum_{v \in \mc C_n} P'_{Y_v} (y) && \text{(channel output)} \\
        &= \frac{1}{|\mc C_n|} \left( \sum_{v \, \in \, \mc C_n \, \setminus \, s_i^{-1} \cdot \mc C_n} \hspace{-0.5cm} P'_{Y_v} (y) + \hspace{-0.5cm} \sum_{v \, \in \, \mc C_n \, \cap \, s_i^{-1} \cdot \mc C_n} \hspace{-0.5cm} P'_{Y_v} (y) \right) \quad \forall y \in \mc Y . \label{eq:second-sum}
    \end{align}

    By~\eqref{eq:first-sum} and~\eqref{eq:second-sum}, the absolute difference of $\widetilde P_{Y'_i} (y)$ and $\widetilde P_Y(y)$ can be upper-bounded as
    \begin{align}
        \left| \widetilde P_{Y'_i} (y) - \widetilde P_Y(y) \right|
        &\le \frac{1}{|\mc C_n|} \left( \sum_{v \, \in \, s_i^{-1} \cdot \mc C_n \, \setminus \, \mc C_n} \hspace{-0.5cm} \left| P'_{ Y_{v} } (y'_i) \right| + \hspace{-0.25cm} \sum_{v \, \in \, \mc C_n \, \setminus \, s_i^{-1} \cdot \mc C_n} \hspace{-0.5cm} \left| P'_{Y_v} (y) \right| \right) && \text{(triangle ineq.)} \\
        &\le \frac{1}{|\mc C_n|} \Big( \left| s_i^{-1} \cdot \mc C_n \setminus \mc C_n \right|+ \left| \mc C_n \setminus s_i^{-1} \cdot \mc C_n \right| \Big) \\
        &\le 2 \frac{\left| \mc B_n \right|}{|\mc C_n|}
        \le 4 \frac{\left| \mc B_n \right|}{|\mc C_n|}
        = \delta_n , \label{eq:lemma-cond-1}
    \end{align}
    where~\eqref{eq:lemma-cond-1} follows because the sets $s_i^{-1} \cdot \mc C_n \setminus \mc C_n$ and $\mc C_n \setminus s_i^{-1} \cdot \mc C_n$ have the same cardinality, and $s_i^{-1} \cdot \mc C_n \setminus \mc C_n$
    is a subset of the set of boundary nodes $\mc B_n \subset \mc V_n$ by the definition of boundary nodes in Definition~\ref{def:approx-sym}.

    Next, we analyze the marginals $\widetilde P_{Y'_{\mc S}}$ and $\widetilde P_{Y'_{\mc T}}$ for $\mc S \sim \mc T$.
    Take any two ordered subsets $S, T$ of $[d]$ with $\mc S \sim \mc T$ with $|\mc S|=|\mc T|=k>0$.
    For any node $v \in \ms C$, denote
    \begin{align}
        \pai{S}{v} &= ( s_{i_1}^{-1} \cdot v \, , \, \ldots \, , \, s_{i_k}^{-1} \cdot v ) , \\
        \pai{T}{v} &= ( s_{j_1}^{-1} \cdot v \, , \, \ldots \, , \, s_{j_k}^{-1} \cdot v ) .
    \end{align}
    By Definition~\ref{def:poset-equiv-ordered-subsets}, there exist nodes
    $v_0, w_0 \in \ms C$ such that $\pai{S}{v_0} = \pai{T}{w_0}$.
    This implies\linebreak ${s_{i_\ell}^{-1} \cdot v_0 = s_{j_\ell}^{-1} \cdot w_0}$ for all $\ell \in [k]$.
    Equivalently, $s_{j_\ell} \cdot s_{i_\ell}^{-1} = w_0 \cdot v_0^{-1}$ for all $\ell \in [k]$.
    This implies
    \begin{equation} \label{eq:relation-S-T-v}
        \pai{S}{v} = \pai{T}{s_{j_\ell} \cdot s_{i_\ell}^{-1} \cdot v} \quad \forall \ell \in [k], \, \forall v \in \ms C .
    \end{equation}

    $\widetilde P_{Y'_{\mc S}}$ can be written as
    \begin{align}
        \widetilde P_{Y'_{\mc S}} (y'_{\mc S}) &= \sum_{x , \overline{y}_{\setminus S}'} \widetilde P_{X , \overline{Y}'} (x , \overline{y}')
        = \sum_{x , \overline{y}_{\setminus S}'} \frac{1}{|\mc C_n|} \sum_{v \in \mc C_n} \hspace{-0.075cm} P'_{ X_v , Y_{\pa v} } (x , \overline{y}') && \text{(marginalization \& \eqref{eq:poset-upper-b-def-P-tilde})} \\
        &= \frac{1}{|\mc C_n|} \sum_{v \in \mc C_n} \sum_{x , \overline{y}_{\setminus S}'} P'_{ X_v , Y_{\pa v} } (x , \overline{y}') \\
        &= \frac{1}{|\mc C_n|} \sum_{v \in \mc C_n} P'_{ Y_{\mathrm{pa}_{\mc S}(v)} } (y'_{\mc S}) \qquad \forall y'_{\mc S} \in \mc Y^{k} . && \text{(marginalization)}
    \end{align}
    Analogously, $P_{Y'_{\mc T}}$ equals
    \begin{align}
        \widetilde P_{Y'_{\mc T}} (y'_{\mc T})
        = \frac{1}{|\mc C_n|} \sum_{v \in \mc C_n} P'_{ Y_{\mathrm{pa}_{\mc T}(v)} } (y'_{\mc T}) \qquad \forall y'_{\mc T} \in \mc Y^{k} .
    \end{align}
    
    By~\eqref{eq:relation-S-T-v}, for each $v \in \mc C_n$,
    whenever $v' \coloneq s_{j_1} \cdot s_{i_1}^{-1} \cdot v$ belongs to the set $\mc C_n$,
    the terms $P'_{ Y_{\mathrm{pa}_{\mc S}(v)} } (y'_{\mc S})$ and $P'_{ Y_{\mathrm{pa}_{\mc T}(v')} } (y'_{\mc S})$ cancel out in $\widetilde P_{Y'_{\mc S}} (y'_{\mc S}) - \widetilde P_{Y'_{\mc T}} (y'_{\mc S})$.
    We thus express $\widetilde P_{Y'_{\mc S}}$ and $\widetilde P_{Y'_{\mc T}}$ as
    \begin{align}
        \widetilde P_{Y'_{\mc S}} (y'_{\mc S}) &= \frac{1}{|\mc C_n|} \sum_{v \in \mc C_n} P'_{ Y_{\mathrm{pa}_{\mc S}(v)} } (y'_{\mc S})
        = \frac{1}{|\mc C_n|} \left( \sum_{\substack{ v \in \mc C_n \\ s_{j_1} \cdot s_{i_1}^{-1} \cdot v \notin \mc C_n } } P'_{ Y_{ \mathrm{pa}_{\mc S}(v) } } (y'_{\mc S}) + \sum_{\substack{ v \in \mc C_n \\ s_{j_1} \cdot s_{i_1}^{-1} \cdot v \in \mc C_n } } P'_{ Y_{ \mathrm{pa}_{\mc S}(v) } } (y'_{\mc S}) \right) ,
    \end{align}
    and
    \begin{equation}
        \widetilde P_{Y'_{\mc T}} (y'_{\mc T})
        = \frac{1}{|\mc C_n|} \left( \sum_{\substack{ w \in \mc C_n \\ s_{i_1} \cdot s_{j_1}^{-1} \cdot w \notin \mc C_n } } P'_{ Y_{ \mathrm{pa}_{\mc T}(w) } } (y'_{\mc T}) + \sum_{\substack{ w \in \mc C_n \\ s_{i_1} \cdot s_{j_1}^{-1} \cdot w \in \mc C_n } } P'_{ Y_{ \mathrm{pa}_{\mc T}(w) } } (y'_{\mc T}) \right) .
    \end{equation}
    Notice that in the two decompositions above, the second sums equal each other when $y'_{\mc S} = y'_{\mc T}$.
    This is because each $v \in \mc C_n$ such that $s_{j_1} \cdot s_{i_1}^{-1} \cdot v \in \mc C_n$ corresponds to a $w \in \mc C_n$ such that $s_{i_1} \cdot s_{j_1}^{-1} \cdot w \in \mc C_n$.
    Therefore,
    \begin{equation}
        \widetilde P_{Y'_{\mc S}} (y'_{\mc S}) - \widetilde P_{Y'_{\mc T}} (y'_{\mc S})
        = \frac{1}{|\mc C_n|} \left( \sum_{\substack{ v \in \mc C_n \\ s_{j_1} \cdot s_{i_1}^{-1} \cdot v \notin \mc C_n } } P'_{ Y_{ \mathrm{pa}_{\mc S}(v) } } (y'_{\mc S}) - \sum_{\substack{ w \in \mc C_n \\ s_{i_1} \cdot s_{j_1}^{-1} \cdot w \notin \mc C_n } } P'_{ Y_{ \mathrm{pa}_{\mc T}(w) } } (y'_{\mc S}) \right) .
    \end{equation}
    We upper-bound the cardinalities of the index sets of the two sums.
    We have
    \begin{align}
        \left| \{ v \in \mc C_n \, \big| \, s_{j_1} \cdot s_{i_1}^{-1} \cdot v \notin \mc C_n \} \right|
        & \ = \ \left| \{ v \in \mc C_n \, \big| \, s_{j_1} \cdot s_{i_1}^{-1} \cdot v \notin \mc C_n \ \text{ and } \ s_{i_1}^{-1} \cdot v \notin \mc C_n \} \right| \\
        & \ \qquad + \left| \{ v \in \mc C_n \, \big| \, s_{j_1} \cdot s_{i_1}^{-1} \cdot v \notin \mc C_n \ \text{ and } \ s_{i_1}^{-1} \cdot v \in \mc C_n \} \right| \\[0.5ex]
        & \ \le \ \left| \{ v \in \mc C_n \, \big| \ \ s_{i_1}^{-1} \cdot v \notin \mc C_n \} \right| \\
        & \ \qquad + \left| \{ v' \in \mc C_n \, \big| \, s_{j_1} \cdot v' \notin \mc C_n \} \right| \qquad \quad (v' = s_{i_1}^{-1} \cdot v) \\[0.5ex]
        & \ \le  2 |\mc B_n| .
    \end{align}
    Analogously, we also have
    $\left| \{ w \in \mc C_n \, \big| \, s_{i_1} \cdot s_{j_1}^{-1} \cdot w \notin \mc C_n \} \right| \le 2 |\mc B_n|$.
    These imply
    \begin{align}
        \left| \widetilde P_{Y'_{\mc S}} (y'_{\mc S}) - \widetilde P_{Y'_{\mc T}} (y'_{\mc S}) \right|
        & \ \le \ \frac{1}{|\mc C_n|} \left( \sum_{\substack{ v \in \mc C_n \\ s_{j_1} \cdot s_{i_1}^{-1} \cdot v \notin \mc C_n } } \left| P'_{ Y_{ \mathrm{pa}_{\mc S}(v) } } (y'_{\mc S}) \right| + \sum_{\substack{ w \in \mc C_n \\ s_{i_1} \cdot s_{j_1}^{-1} \cdot w \notin \mc C_n } } \left| P'_{ Y_{ \mathrm{pa}_{\mc T}(w) } } (y'_{\mc S}) \right| \right) && \text{(triangle ineq.)} \\
        & \ \le \ 4 \frac{|\mc B_n|}{|\mc C_n|} . \label{eq:part2}
    \end{align}

    Inequalities~\eqref{eq:lemma-cond-1} and~\eqref{eq:part2}, according to the definition of $D_\epsilon$ in~\eqref{eq:poset-upper-b-def-D-eps}, imply that $\widetilde P_{X , \overline{Y}'} \in D_{\delta_n}$.

    The conditional mutual information $I(X;Y \mid \ov Y')$ is a concave function of $P_{X , \overline{Y}'}$, as the channel constants $Q_{Y \mid X , \ov Y'}$ are fixed. Hence, due to $\widetilde P^{(n)}_{X , \overline{Y}'}$ being defined through the averaging in~\eqref{eq:poset-upper-b-def-P-tilde}, by Jensen's inequality, we get
    \begin{equation} \label{eq:poset-upper-b-jensen}
        \frac{1}{|\mc C_n|} \sum_{v \in \mc C_n} I'( X_v ; Y_v \mid Y_{\pa v} )
        \quad \le \quad
        I(X;Y \mid \ov Y') \, \Bigg|_{ P_{X , \overline{Y}'} = \widetilde P^{(n)}_{X , \overline{Y}'} }
    \end{equation}
    where the mutual information terms $I'$ on the LHS are induced by the variables $\left\{ P'_{X_v , Y_{\pa v}} \right\}_{v \in \mc C_n}$, and the channel constants.
    \eqref{eq:poset-upper-b-jensen} implies that for any feasible $\left\{ P'_{X_v , Y_{\pa v}} \right\}_{v \in \mc C_n}$ for the LHS of~\eqref{eq:poset-upper-b-instance}, there is a feasible $\widetilde P^{(n)}_{X , \overline{Y}'} \in D_{\delta_n}$ for the RHS with induced objective value at least as large.
    Therefore, maximizing the LHS over all feasible $\left\{ P'_{X_v , Y_{\pa v}} \right\}_{v \in \mc C_n}$ and maximizing the RHS over all $P_{X , \ov Y'} \in D_{\delta_n}$, the inequality in~\eqref{eq:poset-upper-b-jensen} will still hold.
    That is, we have
    \begin{equation} \label{eq:poset-upper-b-instance-2}
        \begin{aligned}
            &\max_{ \left\{ P_{X_v , Y_{\pa v}} \right\}_{v \in \mc C_n} } \ \frac{1}{|\mc C_n|} \sum_{v \in \mc C_n} I( X_v ; Y_v \mid Y_{\pa v} )
            \quad \le \quad
            \max_{ P_{X , \ov Y'} \in D_{\delta_n} } I(X;Y \mid \ov Y') . \\[-1.5ex]
            & \hspace{1cm} \textup{s.t.} \qquad \mathrm{LE^{(n)}_{\text{c}}}
        \end{aligned}
    \end{equation}

    Each $D_{\delta_n}$ is a compact set by the definition in~\eqref{eq:poset-upper-b-def-D-eps}. On the other hand, we know that the sequence $\{\delta_n\}_n$ converges to $0$, and thus has a monotonically decreasing subsequence that converges to $0$ as well. Denote such a subsequence by $\{\delta_{n_k}\}_{k}$. Each $D_{\delta_{n_k}}$ is a compact set, and the sequence of sets $\{D_{\delta_{n_k}}\}_k$ is therefore a decreasing sequence of compact sets, and hence has nonempty limit $D_0$, given by
    \begin{equation} \label{eq:D0}
        D_0 \coloneq \left\{ \left. P_{X , \ov Y'} \, \right| \, P_Y = P_{Y'_i} \ \forall i \, , \ \ P_{Y'_{\mc S}} = P_{Y'_{\mc T}} \ \forall \mc S \sim \mc T \right\} ,
    \end{equation}
    which is precisely the constraint set of the single-letter expression on the RHS of~\eqref{eq:poset-converse}. Therefore, taking $n$ to infinity on both sides of~\eqref{eq:poset-upper-b-instance-2}, we obtain
    \begin{align}
        \lim_{n\to\infty} \max_{ \left\{ P_{X_v , Y_{\pa v}} \right\}_{v \in \mc C_n} } \ \frac{1}{|\mc C_n|} \sum_{v \in \mc C_n} I( X_v ; Y_v \mid Y_{\pa v} )
        \quad &\le \quad
        \lim_{n\to\infty} \, \max_{ P_{X , \ov Y'} \in D_{\delta_n} } I(X;Y \mid \ov Y') \\
        \textup{s.t.} \qquad \mathrm{LE^{(n)}_{\text{c}}} \hspace{4.25cm} &= \quad \max_{ P_{X , \ov Y'} \in D_0 } I(X;Y \mid \ov Y') \\
        &= \quad \max_{ P_{X , \ov Y'} } \, I(X;Y \mid \ov Y') \\[-1.5ex]
        & \hspace{1.1cm} \textup{s.t. } P_Y = P_{Y'_i} \ \forall i , \notag \\[-2ex]
        & \hspace{1.575cm} P_{Y'_{\mc S}} = P_{Y'_{\mc T}} \ \forall \mc S \sim \mc T \notag
    \end{align}
    as desired. This concludes the proof of Theorem~\ref{thm:poset-single-letter}.
\end{silentproof}


\section{} \label{sec:proof-alpha-ge-0.5}

Assume $P_Y = P_{Y'_1} = P_{Y'_2}$ and $\alpha \ge 0.5$.
We have
\begin{align}
    P_Y(1) &= P_M(1) \cdot (1-\alpha) \\
    &= (1-\alpha) \cdot \left[ P_{Y'_1,Y'_2} (1,1) + P_{Y'_1,Y'_2} (1,0) \cdot P_{X \mid Y'_1,Y'_2} (1 \mid 1,0) + P_{Y'_1,Y'_2} (0,1) \cdot P_{X \mid Y'_1,Y'_2} (1 \mid 0,1) \right] \label{eq:z-proof-0} \\
    &\le (1-\alpha) \cdot \left[ P_{Y'_1,Y'_2} (1,1) + P_{Y'_1,Y'_2} (1,0) + P_{Y'_1,Y'_2} (0,1) \right] . \label{eq:z-proof-1}
\end{align}
At the same time, due to $P_Y = P_{Y'_1} = P_{Y'_2}$, we have
\begin{align}
    P_Y(1) &= P_{Y'_1}(1) = P_{Y'_1,Y'_2} (1,1) + P_{Y'_1,Y'_2} (1,0) \label{eq:z-proof-2} \\
    &= P_{Y'_2}(1) = P_{Y'_1,Y'_2} (1,1) + P_{Y'_1,Y'_2} (0,1) . \label{eq:z-proof-3}
\end{align}
By~\eqref{eq:z-proof-1}--\eqref{eq:z-proof-3} we obtain
\begin{align}
    P_{Y'_1,Y'_2} (1,1) + P_{Y'_1,Y'_2} (1,0) &\le (1-\alpha) \cdot \left[ P_{Y'_1,Y'_2} (1,1) + P_{Y'_1,Y'_2} (1,0) + P_{Y'_1,Y'_2} (0,1) \right] , \label{eq:z-proof-4} \\
    P_{Y'_1,Y'_2} (1,1) + P_{Y'_1,Y'_2} (0,1) &\le (1-\alpha) \cdot \left[ P_{Y'_1,Y'_2} (1,1) + P_{Y'_1,Y'_2} (1,0) + P_{Y'_1,Y'_2} (0,1) \right] . \label{eq:z-proof-5}
\end{align}
Summing both sides of these two inequalities yield
\begin{equation} \label{eq:z-proof-6}
    2 \cdot P_{Y'_1,Y'_2} (1,1) + P_{Y'_1,Y'_2} (1,0) + P_{Y'_1,Y'_2} (0,1)
    \le 2(1-\alpha) \cdot \left[ P_{Y'_1,Y'_2} (1,1) + P_{Y'_1,Y'_2} (1,0) + P_{Y'_1,Y'_2} (0,1) \right] .
\end{equation}
\paragraph*{For $\alpha > 0.5$}
\eqref{eq:z-proof-6} implies $P_{Y'_1,Y'_2} (1,1) = P_{Y'_1,Y'_2} (1,0) = P_{Y'_1,Y'_2} (0,1) = 0$.
Therefore,\linebreak ${(Y'_1,Y'_2)=(1,1)}$ deterministically, which implies $M=1$ deterministically. This means $Y$ is independent of $X$, and thus $I(X;Y \mid Y'_1,Y'_2) = 0$.
\paragraph*{For $\alpha = 0.5$}
\eqref{eq:z-proof-6} implies $P_{Y'_1,Y'_2} (1,1) = 0$.
By~\eqref{eq:z-proof-4} and~\eqref{eq:z-proof-5}, we get\linebreak ${P_{Y'_1,Y'_2} (1,0) = P_{Y'_1,Y'_2} (1,0) = p}$.
Then by~\eqref{eq:z-proof-2} 
we have that $P_Y(1) = p$ as well.
These together with~\eqref{eq:z-proof-0} force $P_{X \mid Y'_1,Y'_2} (1 \mid 1,0) = P_{X \mid Y'_1,Y'_2} (1 \mid 0,1) = 1$.
All in all, $P_{Y'_1,Y'_2} (1,1) = 0$ and\linebreak ${P_{X \mid Y'_1,Y'_2} (1 \mid 1,0) = P_{X \mid Y'_1,Y'_2} (1 \mid 0,1) = 1}$ imply that $X$ cannot affect $M$ or $Y$. Consequently,\linebreak ${I(X;Y \mid Y'_1,Y'_2) = 0}$.

\newpage


\begin{thebibliography}{10}
\providecommand{\url}[1]{#1}
\csname url@samestyle\endcsname
\providecommand{\newblock}{\relax}
\providecommand{\bibinfo}[2]{#2}
\providecommand{\BIBentrySTDinterwordspacing}{\spaceskip=0pt\relax}
\providecommand{\BIBentryALTinterwordstretchfactor}{4}
\providecommand{\BIBentryALTinterwordspacing}{\spaceskip=\fontdimen2\font plus
\BIBentryALTinterwordstretchfactor\fontdimen3\font minus \fontdimen4\font\relax}
\providecommand{\BIBforeignlanguage}[2]{{%
\expandafter\ifx\csname l@#1\endcsname\relax
\typeout{** WARNING: IEEEtran.bst: No hyphenation pattern has been}%
\typeout{** loaded for the language `#1'. Using the pattern for}%
\typeout{** the default language instead.}%
\else
\language=\csname l@#1\endcsname
\fi
#2}}
\providecommand{\BIBdecl}{\relax}
\BIBdecl

\bibitem{Shannon1948}
C.~E. Shannon, ``A mathematical theory of communication,'' \emph{The Bell System Technical Journal}, vol.~27, no.~3, pp. 379--423, Oct. 1948.

\bibitem{Gallager1962}
R.~Gallager, ``Low-density parity-check codes,'' \emph{IRE Transactions on Information Theory}, vol.~8, no.~1, pp. 21--28, Jan. 1962.

\bibitem{Berrou1993}
C.~Berrou, A.~Glavieux, and P.~Thitimajshima, ``Near {S}hannon limit error-correcting coding and decoding: Turbo-codes.~1,'' in \emph{Proceedings of ICC '93 - IEEE International Conference on Communications}, vol.~2, May 1993, pp. 1064--1070 vol.2.

\bibitem{Arikan2009}
E.~Arikan, ``Channel polarization: A method for constructing \mbox{capacity-achieving} codes for symmetric binary-input memoryless channels,'' \emph{IEEE Transactions on Information Theory}, vol.~55, no.~7, pp. 3051--3073, Jul. 2009.

\bibitem{Khina2019}
A.~Khina, V.~Kostina, A.~Khisti, and B.~Hassibi, ``Tracking and control of gauss–markov processes over packet-drop channels with acknowledgments,'' \emph{IEEE Transactions on Control of Network Systems}, vol.~6, no.~2, pp. 549--560, Jun. 2019.

\bibitem{Han2023}
B.~Han, O.~Sabag, V.~Kostina, and B.~Hassibi, ``Coded kalman filtering over gaussian channels with feedback,'' in \emph{2023 59th Annual Allerton Conference on Communication, Control, and Computing (Allerton)}, Sep. 2023, pp. 1--8.

\bibitem{Han2024}
B.~Han, V.~Kostina, B.~Hassibi, and O.~Sabag, ``Coded kalman filtering over mimo gaussian channels with feedback,'' in \emph{2024 IEEE International Symposium on Information Theory (ISIT)}, Jul. 2024, pp. 3261--3266.

\bibitem{Shayevitz2011posterior}
O.~Shayevitz and M.~Feder, ``Optimal feedback communication via posterior matching,'' \emph{IEEE Transactions on Information Theory}, vol.~57, no.~3, pp. 1186--1222, Mar. 2011.

\bibitem{Horstein1963}
M.~Horstein, ``Sequential transmission using noiseless feedback,'' \emph{IEEE Transactions on Information Theory}, vol.~9, no.~3, pp. 136--143, Jan. 1963.

\bibitem{Elias1956}
P.~Elias, ``Channel capacity without coding,'' Massachusetts Institute of Technology, Research Laboratory of Electronics, Cambridge, MA, USA, Quarterly Progress Report~54, Oct. 1956, reprinted in \textit{Lectures on Communication System Theory}, E.~Baghdady, Ed. New York, NY, USA: McGraw-Hill, 1961, pp.~363--366.

\bibitem{Dobrushin1963}
R.~L. Dobrushin, ``General formulation of {S}hannon's main theorem in information theory,'' \emph{American Mathematical Society Translations (Series~2)}, vol.~33, pp. \mbox{323--438}, 1963.

\bibitem{Massey1990}
J.~Massey, ``Causality, feedback and directed information,'' in \emph{Proceedings of the International Symposium on Information Theory and Its Applications (ISITA)}, vol.~2, Nov. 1990.

\bibitem{Tatikonda2009}
S.~Tatikonda and S.~Mitter, ``The capacity of channels with feedback,'' \emph{IEEE Transactions on Information Theory}, vol.~55, no.~1, pp. 323--349, Jan. 2009.

\bibitem{Boche2020FSC}
H.~Boche, R.~F. Schaefer, and H.~V. Poor, ``Shannon meets {T}uring: Non-computability and non-approximability of the finite state channel capacity,'' \emph{Communications in Information and Systems}, vol.~20, no.~2, pp. 81--116, Aug. 2020.

\bibitem{Grigorescu2024feedback}
A.~Grigorescu, H.~Boche, R.~F. Schaefer, and H.~V. Poor, ``Capacity of finite state channels with feedback: Algorithmic and optimization theoretic properties,'' \emph{IEEE Transactions on Information Theory}, vol.~70, no.~8, pp. 5413--5426, Aug. 2024.

\bibitem{Butman}
S.~Butman, ``A general formulation of linear feedback communication systems with solutions,'' \emph{IEEE Transactions on Information Theory}, vol.~15, no.~3, pp. 392--400, May 1969.

\bibitem{CoverPombra1989}
T.~Cover and S.~Pombra, ``Gaussian feedback capacity,'' \emph{IEEE Transactions on Information Theory}, vol.~35, no.~1, pp. 37--43, Jan. 1989.

\bibitem{Elia}
N.~Elia, ``When {Bode} meets {Shannon}: control-oriented feedback communication schemes,'' \emph{IEEE Transactions on Automatic Control}, vol.~49, no.~9, pp. 1477--1488, Sep. 2004.

\bibitem{Kim-Stationary-Gaussian}
Y.-H. Kim, ``Feedback capacity of stationary {Gaussian} channels,'' \emph{IEEE Transactions on Information Theory}, vol.~56, no.~1, pp. 57--85, Jan. 2010.

\bibitem{Gattami-Gaussian}
A.~Gattami, ``Feedback capacity of {Gaussian} channels revisited,'' \emph{IEEE Transactions on Information Theory}, vol.~65, no.~3, pp. 1948--1960, Mar. 2019.

\bibitem{Sabag-MIMO-Gaussian}
O.~Sabag, V.~Kostina, and B.~Hassibi, ``Feedback capacity of {MIMO Gaussian} channels,'' \emph{IEEE Transactions on Information Theory}, vol.~69, no.~10, pp. 6121--6136, Oct. 2023.

\bibitem{trapdoor}
H.~Permuter, P.~Cuff, B.~Van~Roy, and T.~Weissman, ``Capacity of the trapdoor channel with feedback,'' \emph{IEEE Transactions on Information Theory}, vol.~54, no.~7, pp. 3150--3165, Jul. 2008.

\bibitem{Ising}
O.~Elishco and H.~Permuter, ``Capacity and coding for the {Ising} channel with feedback,'' \emph{IEEE Transactions on Information Theory}, vol.~60, no.~9, pp. 5138--5149, Sep. 2014.

\bibitem{permuter2014post}
H.~H. Permuter, H.~Asnani, and T.~Weissman, ``Capacity of a {POST} channel with and without feedback,'' \emph{IEEE Transactions on Information Theory}, vol.~60, no.~10, pp. 6041--6057, Sep. 2014.

\bibitem{NOST}
E.~Shemuel, O.~Sabag, and H.~H. Permuter, ``The feedback capacity of {Noisy Output Is the STate (NOST)} channels,'' \emph{IEEE Transactions on Information Theory}, vol.~68, no.~8, pp. 5044--5059, Aug. 2022.

\bibitem{huleihel2024capacity}
B.~Huleihel, O.~Sabag, H.~H. Permuter, and V.~Kostina, ``Capacity of finite-state channels with delayed feedback,'' \emph{IEEE Transactions on Information Theory}, vol.~70, no.~1, pp. 16--29, Jan. 2024.

\bibitem{Sabag-no-consec-ones}
O.~Sabag, H.~H. Permuter, and N.~Kashyap, ``The feedback capacity of the binary erasure channel with a no-consecutive-ones input constraint,'' \emph{IEEE Transactions on Information Theory}, vol.~62, no.~1, pp. 8--22, Jan. 2016.

\bibitem{Sabag-BIBO}
------, ``Feedback capacity and coding for the {BIBO} channel with a no-repeated-ones input constraint,'' \emph{IEEE Transactions on Information Theory}, vol.~64, no.~7, pp. 4940--4961, Jul. 2018.

\bibitem{Peled-0-k-RLL}
O.~Peled, O.~Sabag, and H.~H. Permuter, ``Feedback capacity and coding for the $(0,k)$-{RLL} input-constrained {BEC},'' \emph{IEEE Transactions on Information Theory}, vol.~65, no.~7, pp. 4097--4114, Jul. 2019.

\bibitem{alajaji1994effect}
F.~Alajaji and T.~Fuja, ``Effect of feedback on the capacity of discrete additive channels with memory,'' in \emph{Proceedings of 1994 IEEE International Symposium on Information Theory}, Jun. 1994, p. 464.

\bibitem{alajaji1995feedback}
F.~Alajaji, ``Feedback does not increase the capacity of discrete channels with additive noise,'' \emph{IEEE Transactions on Information Theory}, vol.~41, no.~2, pp. 546--549, Mar. 1995.

\bibitem{song2018capacity}
L.~Song, F.~Alajaji, and T.~Linder, ``Capacity of burst noise-erasure channels with and without feedback and input cost,'' \emph{IEEE Transactions on Information Theory}, vol.~65, no.~1, pp. 276--291, Jan. 2019.

\bibitem{CoverThomas}
T.~M. Cover and J.~A. Thomas, \emph{Elements of Information Theory}, 2nd~ed., ser. Wiley Series in Telecommunications and Signal Processing.\hskip 1em plus 0.5em minus 0.4em\relax Hoboken, NJ, USA: John Wiley \& Sons, 2006.

\bibitem{Mushkin1989}
M.~Mushkin and I.~Bar-David, ``Capacity and coding for the \mbox{Gilbert-Elliott} channels,'' \emph{IEEE Transactions on Information Theory}, vol.~35, no.~6, pp. 1277--1290, Nov. 1989.

\bibitem{Huleihel-upper-bound}
B.~Huleihel, O.~Sabag, H.~H. Permuter, N.~Kashyap, and S.~Shamai~(Shitz), ``Computable upper bounds on the capacity of \mbox{finite-state} channels,'' \emph{IEEE Transactions on Information Theory}, vol.~67, no.~9, pp. 5674--5692, Sep. 2021.

\bibitem{Yang2005}
S.~Yang, A.~Kavcic, and S.~Tatikonda, ``Feedback capacity of finite-state machine channels,'' \emph{IEEE Transactions on Information Theory}, vol.~51, no.~3, pp. 799--810, Mar. 2005.

\bibitem{Sabag-upper-bound}
O.~Sabag, H.~H. Permuter, and H.~D. Pfister, ``A single-letter upper bound on the feedback capacity of unifilar finite-state channels,'' \emph{IEEE Transactions on Information Theory}, vol.~63, no.~3, pp. 1392--1409, Mar. 2017.

\bibitem{Boyd2004convex}
S.~Boyd and L.~Vandenberghe, \emph{Convex optimization}.\hskip 1em plus 0.5em minus 0.4em\relax Cambridge, U.K.: Cambridge University Press, 2004.

\bibitem{chen2005capacity}
J.~Chen and T.~Berger, ``The capacity of finite-state {Markov} channels with feedback,'' \emph{IEEE Transactions on Information Theory}, vol.~51, no.~3, pp. 780--798, Mar. 2005.

\bibitem{berger-2d-ising}
T.~Berger and F.~Bonomi, ``Capacity and zero-error capacity of {Ising} channels,'' \emph{IEEE Transactions on Information Theory}, vol.~36, no.~1, pp. 173--180, Jan. 1990.

\bibitem{ReRAM}
G.~Song, K.~Cai, Y.~Li, and K.~A. Schouhamer~Immink, ``Maximum achievable rate of resistive random-access memory channels by mutual information spectrum analysis,'' \emph{IEEE Transactions on Information \mbox{Theory}}, vol.~69, no.~5, pp. \mbox{2808--2819}, May 2023.

\bibitem{Blackwell1961}
D.~Blackwell, ``Information theory,'' \emph{Modern Mathematics for the Engineer: Second Series}, pp. 183--193, 1961.

\bibitem{berger2002lec}
T.~Berger, ``The generalized {Shannon-Blackwell} billiard ball channel,'' in \emph{Lecture 2 of CSL Distinguished Visiting Professorship}, {Urbana}, IL: Univorsity of Illinois, Apr. 22, 2002, "Information Theory Invades Biology".

\bibitem{Tatikonda2000thesis}
S.~C. Tatikonda, ``Control under communication constraints,'' Ph.D. dissertation, Massachusetts Institute of Technology, 2000.

\bibitem{bonomi1985thesis}
F.~Bonomi, ``Problems in the information theory of random fields,'' Ph.D. dissertation, Cornell University, 1985.

\bibitem{raymond2018symmetric}
A.~Raymond, J.~Saunderson, M.~Singh, and R.~R. Thomas, ``Symmetric sums of squares over $k$-subset hypercubes,'' \emph{Mathematical Programming}, vol. 167, no.~2, pp. 315--354, Feb. 2018.

\bibitem{raymond2018symmetry}
A.~Raymond, M.~Singh, and R.~R. Thomas, ``Symmetry in {Tur{\'a}n} sums of squares polynomials from flag algebras,'' \emph{Algebraic Combinatorics}, vol.~1, no.~2, pp. 249--274, Mar. 2018.

\bibitem{lovasz2012large}
L.~Lov{\'a}sz, \emph{Large networks and graph limits}.\hskip 1em plus 0.5em minus 0.4em\relax Providence, RI: American Mathematical Society, 2012, vol.~60.

\bibitem{brosch_thesis}
D.~Brosch, ``\BIBforeignlanguage{English}{Symmetry reduction in convex optimization with applications in combinatorics},'' Ph.D. dissertation, Tilburg University, 2022.

\bibitem{huber2021positive}
F.~Huber, ``{Positive maps and trace polynomials from the symmetric group},'' \emph{Journal of Mathematical Physics}, vol.~62, no.~2, p. 022203, Feb. 2021.

\bibitem{klep2018positive}
I.~Klep, {\v{S}}.~{\v{S}}penko, and J.~Vol{\v{c}}i{\v{c}}, ``Positive trace polynomials and the universal {Procesi–Schacher} conjecture,'' \emph{Proceedings of the London Mathematical Society}, vol. 117, no.~6, pp. 1101--1134, Jun. 2018.

\bibitem{klep2022optimization}
I.~Klep, V.~Magron, and J.~Vol{\v{c}}i{\v{c}}, ``Optimization over trace polynomials,'' in \emph{Annales Henri Poincar{\'e}}, vol.~23, Jan. 2022, pp. 67--100.

\bibitem{lasry2007mean}
J.-M. Lasry and P.-L. Lions, ``Mean field games,'' \emph{Japanese Journal of Mathematics}, vol.~2, no.~1, pp. 229--260, Mar. 2007.

\bibitem{riener2013exploiting}
C.~Riener, T.~Theobald, L.~J. Andr{\'e}n, and J.~B. Lasserre, ``Exploiting symmetries in {SDP}-relaxations for polynomial optimization,'' \emph{Mathematics of Operations Research}, vol.~38, no.~1, pp. 122--141, Feb. 2013.

\bibitem{levin2023free}
E.~Levin and V.~Chandrasekaran, ``Free descriptions of convex sets,'' \emph{arXiv preprint arXiv:2307.04230}, Jun. 2024.

\bibitem{Szpilrajn1930}
E.~Szpilrajn, ``Sur l'extension de l'ordre partiel,'' \emph{Fundamenta Mathematicae}, vol.~16, pp. 386--389, 1930.

\bibitem{Cayley1878}
A.~Cayley, ``Desiderata and suggestions: No.\ 2—the theory of groups: Graphical representation,'' \emph{American Journal of Mathematics}, vol.~1, no.~2, pp. 174--176, 1878.

\bibitem{Lalley_MC_LecNotes}
S.~P. Lalley, ``Markov chains: {B}asic theory,'' {L}ecture notes for Statistics 312, University of Chicago, 2016. Available: \url{https://galton.uchicago.edu/~lalley/Courses/312/MarkovChains.pdf}.

\bibitem{widder2009fixed}
A.~Widder, ``Fixed point theorems for set-valued maps,'' \emph{Institute for Analysis and Scientific Computing, Vienna University of Technology}, Jul. 2009.

\bibitem{kohlberg1982contraction}
E.~Kohlberg and J.~W. Pratt, ``The contraction mapping approach to the {Perron--Frobenius} theory: Why {Hilbert's} metric?'' \emph{Mathematics of Operations Research}, vol.~7, no.~2, pp. 198--210, May 1982.

\end{thebibliography}


\end{document}